\documentclass[a4paper,UKenglish,autoref,cleverref,numberwithinsect,pdfa,thm-restate]{lipics-v2021}
\usepackage[utf8]{inputenc}
\RequirePackage[T1]{fontenc}
\usepackage{amssymb}
\usepackage{etoolbox}
\usepackage{ifmtarg}
\usepackage{tikz}
\usepackage{xspace, soul, enumerate}

\usepackage{algorithm}
\usepackage[noend]{algpseudocode}
\usepackage{mathtools}
\usepackage{hyperref}

\usepackage{comment}

\newif\ifcomments
\newif\ifchanges

\usepackage{xargs,leftidx}  %

\theoremstyle{definition}
\newtheorem{defi}[theorem]{Definition}

\newcommand{\myappendix}{appendix}

\newcommand{\mtext}[1]{\textsc{#1}}

\newcommand{\N}{\ensuremath{\mathbb{N}}}

\newcommand{\bigO}{\ensuremath{\mathcal{O}}}

\makeatletter 
\newcommand{\ut}[4]{
  \@ifmtarg{#4}{t^{#1}_{#2}(#3) }{t^{#1}_{#2}(#3; #4)}
}

\newcommand{\ite}[3]{
  \@ifmtarg{#1}{
    \mtext{ITE}
   }{
    \mtext{ITE}\text{$(#1,#2,#3)$}  
  }
}
\makeatother

\newcommand  {\myclass} [1]  {\ensuremath{\textsf{\upshape #1}}}

\newcommand{\StaClass}[1]{\myclass{#1}\xspace}

\newcommand{\DynClass}[1]{\myclass{Dyn#1}\xspace}

\newcommand  {\myproblem} [1] {\normalfont{\textsc{#1}}\xspace}

\newcommand{\CQ}[1][]{\StaClass{CQ}}
\newcommand{\UCQ}[1][]{\StaClass{UCQ}}
\newcommand{\CQneg}[1][]{\StaClass{CQ\ensuremath{^{\mneg}}}}
\newcommand{\UCQneg}[1][]{\StaClass{UCQ\ensuremath{^{\mneg}}}}

\newcommand{\mneg}{\neg} %

\newcommand{\DynFO}{\DynClass{FO}}

\newenvironment{proofsketch}{\begin{proof}[Proof sketch]}{\end{proof}}

\newenvironment{proofof}[1]{\begin{proof}[Proof (of #1).]}{\end{proof}}

\providecommand {\calA}      {{\mathcal A}\xspace}
\providecommand {\calB}      {{\mathcal B}\xspace}

\algnewcommand\algorithmiconchange{\textbf{on change}}
\algnewcommand\algorithmiconquery{\textbf{on query}}
\algnewcommand\algorithmicupdate{\textbf{update}}
\algnewcommand\algorithmicat{\textbf{at}}
\algnewcommand\algorithmicby{\textbf{by}}
\algnewcommand\algorithmicpardo{\textbf{pardo}}
\algnewcommand\algorithmicwhere{\textbf{where}}
\algnewcommand\algorithmicwith{\textbf{with}}

\algnewcommand\algorithmicunique{\textbf{unique}}
\algnewcommand\algorithmicmin{\textbf{min}}
\algnewcommand\algorithmicmax{\textbf{max}}

\algnewcommand{\False}{\textbf{false}}
\algnewcommand{\True}{\textbf{true}}
\algnewcommand{\To}{\textbf{to}}
\algnewcommand{\Select}[4]{\State \ensuremath{#1 \gets #2(#3 \mid #4)}}
\algnewcommand{\Unique}[3]{\Select{#1}{\algorithmicunique}{#2}{#3}}
\algnewcommand{\Min}[3]{\Select{#1}{\algorithmicmin}{#2}{#3}}
\algnewcommand{\Max}[3]{\Select{#1}{\algorithmicmax}{#2}{#3}}

\algblockdefx{OnChangeAtWhereBy}{EndOnChangeAtWhereBy}%
    [4]{\algorithmiconchange\ #1\ \algorithmicupdate\ #2 \algorithmicat\ #3, \algorithmicforall\ #4, \algorithmicby:}%
    {\algorithmicend\ \algorithmiconchange}
    
\algblockdefx{UpdateAtWhereBy}{EndUpdateAtWhereBy}%
    [3]{\algorithmicupdate\ #1 \algorithmicat\ #2, \algorithmicforall\ #3, \algorithmicby:}%
    {\algorithmicend\ \algorithmiconchange}    

    \algblockdefx{UpdateAtBy}{EndUpdateAtWhereBy}%
    [2]{\algorithmicupdate\ #1 \algorithmicat\ #2, \algorithmicby:}%
    {\algorithmicend\ \algorithmiconchange}    
    
\algblockdefx{OnQuery}{EndOnQuery}%
    [1]{\algorithmiconquery\ #1:}%
    {\algorithmicend\ \algorithmiconchange}    

\algblockdefx{OnChange}{EndOnChange}%
    [1]{\algorithmiconchange\ #1}%
    {\algorithmicend\ \algorithmiconchange}

    \algblockdefx{OnChangeWith}{EndOnChangeWith}%
    [2]{\algorithmiconchange\ #1 \algorithmicwith\ #2}%
    {\algorithmicend\ \algorithmiconchange}
 
\algblockdefx{With}{EndWith}%
    {\algorithmicwith}%
    {\algorithmicend\ \algorithmiconchange}    

\algblockdefx{AtWhereBy}{EndAtWhereBy}%
    [2]{\algorithmicat\ #1\ \algorithmicwhere\ #2\ \algorithmicby:}%
    {\algorithmicend\ \algorithmiconchange}

\algblockdefx{By}{EndBy}%
    [0]{\algorithmicby:}%
    {\algorithmicend\ \algorithmiconchange}

\algblockdefx{Parfor}{EndFor}%
    [1]{\algorithmicfor\ #1\ \algorithmicpardo}%
    {\algorithmicend\ \algorithmicfor}

\algblockdefx{Parforall}{EndFor}%
    [1]{\algorithmicforall\ #1\ \algorithmicpardo}%
    {\algorithmicend\ \algorithmicfor}

\makeatletter
\ifthenelse{\equal{\ALG@noend}{t}}%
    {
    		\algtext*{EndOnChangeAtWhereBy}
        \algtext*{EndOnChange}
        \algtext*{EndOnQuery}
        \algtext*{EndAtWhereBy}
        \algtext*{EndBy}
        \algtext*{EndFor}
        \algtext*{EndWith}
        \algtext*{EndOnChangeWith}
        \algtext*{EndUpdateAtWhereBy}
    }
    {}%
\makeatother

\newcommand{\neps}[1][]{\ensuremath{\bigO(n^{#1\epsilon})}\xspace}
\newcommand{\ntheta}[1][]{\ensuremath{\bigO(n^{#1\theta})}\xspace}

\newcommand{\blank}{\sqcup}

\providecommand{\nc}{\newcommand}

\newcommand{\child}{\ensuremath{\textsc{child}}\xspace}
\newcommand{\childindex}{\ensuremath{\textsc{child-index}}\xspace}
\newcommand{\nchildren}{\ensuremath{\#\textsc{children}}\xspace}
\newcommand{\nsiblings}{\ensuremath{\#\textsc{lsiblings}}\xspace}
\newcommand{\parent}{\ensuremath{\textsc{parent}}\xspace}
\newcommand{\lsibling}{\ensuremath{\textsc{left-sibling}}\xspace}
\newcommand{\rsibling}{\ensuremath{\textsc{right-sibling}}\xspace}
\newcommand{\fchild}{\ensuremath{\textsc{first-child}}\xspace}
\newcommand{\lchild}{\ensuremath{\textsc{last-child}}\xspace}
\newcommand{\anc}{\ensuremath{\textsc{anc}}\xspace}
\newcommand{\ancself}{\ensuremath{\textsc{anc-or-self}}\xspace}
\newcommand{\ancindex}{\ensuremath{\textsc{anc-index}}\xspace}
\newcommand{\ancchild}{\ensuremath{\textsc{anc-child}}\xspace}

\newcommand{\zsize}{\ensuremath{\textsc{zsize}}\xspace}
\newcommand{\ntops}{\ensuremath{\#\textsc{tops}}\xspace}
\newcommand{\nzleft}{\ensuremath{\#\textsc{zleft}}\xspace}
\newcommand{\nzdesc}{\ensuremath{\#\textsc{zdesc}}\xspace}
\newcommand{\numdesc}{\ensuremath{\#\textsc{desc}}\xspace}

\newcommand{\evaluatestate}{\ensuremath{\textsc{evaluate-state}}}
\newcommand{\evaluatesequence}{\ensuremath{\textsc{evaluate-sequence}}}
\newcommand{\evaluatepath}{\ensuremath{\textsc{evaluate-path}}}

\newcommand{\lca}{\ensuremath{\textsc{lca}}}

\newcommand{\myoperation}[1]{\normalfont{\texttt{#1}}\xspace}

\newcommand{\zleft}{\ensuremath{\textsc{left}}}
\newcommand{\zright}{\ensuremath{\textsc{right}}}
\newcommand{\zlower}{\ensuremath{\textsc{lower}}}
\newcommand{\zupper}{\ensuremath{\textsc{upper}}}

   \newcommand{\status}{\textnormal{status}}

\nc{\Init}{\myoperation{Init}}
\nc{\Query}{\myoperation{Query}}
\nc{\Relabel}{\myoperation{Relabel}}
\nc{\AddChild}{\myoperation{AddChild}}
\nc{\DeleteNode}{\myoperation{DeleteNode}}
\nc{\InsertPositionBefore}{\myoperation{InsertPositionBefore}}
\nc{\InsertPositionAfter}{\myoperation{InsertPositionAfter}}
\nc{\DeletePosition}{\myoperation{DeletePosition}}

\newcommand{\RegTree}{\myproblem{RegTree}}
\newcommand{\RegTreeM}{\myproblem{RegTree\ensuremath{^-}}}
\newcommand{\CFL}{\myproblem{CFL}}
\newcommand{\VPL}{\myproblem{VPL}}

\newcommandx{\tree}[3][1,3]{\ensuremath{\leftidx{^{#1}}{t^{#2}_{#3}}}}

   \newcommand{\run}{\textnormal{run}}

\newcommand{\dstate}{\ensuremath{\delta_{\text{state}}}\xspace}
\newcommand{\dstack}{\ensuremath{\delta_{\text{stack}}}\xspace}
\newcommand{\dempty}{\ensuremath{\delta_{\text{empty}}}\xspace}
\newcommand{\dall}{\ensuremath{\hat{\delta}}\xspace}
   \newcommand{\pushpos}{\ensuremath{\textnormal{push-pos}}\xspace}
  \newcommand{\poppos}{\ensuremath{\textnormal{pop-pos}}\xspace}
 \newcommand{\popposi}{\ensuremath{\textnormal{pop-pos}^{(1)}}\xspace}
 \newcommand{\popposii}{\ensuremath{\textnormal{pop-pos}^{(2)}}\xspace}
  \newcommand{\vpopstate}{\ensuremath{\textnormal{v-pop-state}}\xspace}
 \newcommand{\topk}[1][k]{\ensuremath{\textnormal{top}_{#1}}\xspace}

\ifcomments
\newcommand{\commentbox}[1]{\noindent\framebox{\parbox{0.98\linewidth}{#1}}}

\newcommand{\acomment}[2]{\ \\ \fbox{\parbox{0.98\linewidth}{{\sc #1}: #2}}}
\newcommand{\mcomment}[2]{{\color{blue}(#1)}\footnote{#1: #2}} %
\else
\newcommand{\commentbox}[1]{}
\newcommand{\mcomment}[2]{}
\newcommand{\acomment}[2]{}
\fi

\ifchanges

\setul{}{0.2mm}
\setstcolor{red}

\else

\fi

 \newcommand{\tsm}[1]{\mcomment{TS}{#1}}

\newcommand{\longversion}[1]{}

\title{On the work of dynamic constant-time parallel algorithms for regular tree
languages and context-free languages}
\titlerunning{Dynamic  parallel algorithms for regular tree
languages and context-free languages}

\author{Jonas Schmidt}{TU Dortmund University, Germany}{jonas2.schmidt@tu-dortmund.de}{}{}
\author{Thomas Schwentick}{TU Dortmund University, Germany}{thomas.schwentick@tu-dortmund.de}{}{}
\author{Jennifer Todtenhoefer}{TU Dortmund University, Germany}{jennifer.todtenhoefer@tu-dortmund.de}{}{}
\authorrunning{J. Schmidt, T. Schwentick and J. Todtenhoefer}
\Copyright{Jonas Schmidt, Thomas Schwentick, and Jennifer Todtenhoefer}

\EventEditors{J\'{e}r\^{o}me Leroux, Sylvain Lombardy, and David Peleg}
\EventNoEds{3}
\EventLongTitle{48th International Symposium on Mathematical Foundations of Computer Science (MFCS 2023)}
\EventShortTitle{MFCS 2023}
\EventAcronym{MFCS}
\EventYear{2023}
\EventDate{August 28 to September 1, 2023}
\EventLocation{Bordeaux, France}
\EventLogo{}
\SeriesVolume{272}
\ArticleNo{83}

\relatedversion{This is the full version of a paper to be presented at MFCS 2023.}

\begin{document}
\maketitle              %
\begin{abstract}
    Previous work on Dynamic Complexity has established that there exist dynamic constant-time parallel algorithms for regular tree
    languages and context-free languages under label or symbol changes. However, these algorithms were
    not developed with the goal to minimise work (or,
    equivalently, the number of processors). In fact, their inspection
    yields the work bounds $\bigO(n^2)$ and $\bigO(n^7)$ per change operation,  respectively.

    In this paper, dynamic algorithms for regular tree languages are
    proposed that generalise the previous algorithms in that they allow
    unbounded node rank and leaf insertions, while improving
    the work bound from   $\bigO(n^2)$ to $\bigO(n^\epsilon)$, for
    arbitrary $\epsilon>0$.

    For context-free languages, algorithms with better work bounds
    (compared with  $\bigO(n^7)$) for restricted classes are
    proposed:  for every $\epsilon>0$ there are such algorithms for deterministic context-free languages with
    work bound $\bigO(n^{3+\epsilon})$ and for visibly pushdown
    languages with work bound  $\bigO(n^{2+\epsilon})$.
\end{abstract}

\section{Introduction}\label{section:introduction}
It has been known for many years that regular and context-free string languages and
regular tree languages are maintainable under symbol changes by means
of dynamic algorithms that are specified by formulas of first-order
logic, that is, in the dynamic class \DynFO
\cite{PatnaikI97,GeladeMS12}. It is also well-known that such
specifications can be turned into parallel algorithms for the CRCW
PRAM model that require only constant time \cite{Immerman12} and
polynomially many processors.

However, an ``automatic'' translation of a ``dynamic program'' of the
$\DynFO$ setting usually yields a
parallel algorithm with large work,  i.e., overall
number of operations performed by all processors.\footnote{We note that in
  the context of constant-time parallel algorithms
 work is within a constant factor of the
 number of processors.}
In the case of regular languages, the dynamic program sketched  in
\cite{PatnaikI97} has a polynomial work bound, in which the exponent
of the polynomial
depends on the number of states of a DFA for the language at hand. 
The dynamic program given in \cite{GeladeMS12} has quadratic work.

Only recently a line of research has started that tries to determine, how efficient such
constant-time dynamic
algorithms can be made with respect to their work. 
It turned out that regular languages can be maintained
with work \neps, for every $\epsilon>0$ \cite{SchmidtSTVZ21}, even
under polylogarithmic numbers of changes \cite{TschirbsVZ23}, and even with logarithmic
work for star-free languages under single changes
\cite{SchmidtSTVZ21} and  polylogarithmic
work under  polylogarithmic changes \cite{TschirbsVZ23}.

For context-free languages the situation is much less clear. The
dynamic algorithms resulting from \cite{GeladeMS12} have an
$\bigO(n^7)$ upper work bound. In  \cite{SchmidtSTVZ21} it was shown
that the Dyck-1 language, i.e., the set of well-bracketed strings with
one bracket type, can be maintained with work $\bigO((\log n)^3)$ and
that Dyck-$k$ languages can be maintained with work
$\bigO(n\log n)$. Here, the factor $n$
is due to the problem to test equality of two substrings of a string.

Most of these results also hold for the query that asks
for membership of a substring in the given language.
For  Dyck languages the upper bounds for substring queries are worse
than the bounds for membership queries: for every
$\epsilon>0$ there exist algorithms for  Dyck-1 and Dyck-$k$ languages
with work bounds \neps and $\neps[1+]$, respectively.

It was also shown in \cite{SchmidtSTVZ21} that there is some context-free language that
can be maintained in constant time with work
$\bigO(n^{\omega-1-\epsilon})$, for any $\epsilon>0$, only if the
$k$-Clique conjecture \cite{AbboudBackurs+2018} fails. Here, $\omega$
is the matrix multiplication exponent, which is known to be smaller
than $2.373$ and conjectured by some to be exactly two according to \cite{Williams2012}.\\

In this paper, we pursue two natural research directions.

\subparagraph*{Regular tree languages.} We first  extend the results on regular string languages to
  regular tree languages. On one hand, this requires to adapt techniques from
  strings to trees. On the other hand, trees offer additional
  types of change operations beyond label changes that might change the structure of the
  tree. More concretely, besides label changes we study insertions of
  new leaves and show that the favourable bounds of
  \cite{SchmidtSTVZ21} for regular string languages still hold. This
  is  the main contribution of this paper. Our algorithms rely on a hierarchical partition of the tree
of constant depth. The main technical challenge is to maintain such a
partition hierarchy under insertion\footnote{For simplicity, we only
  consider insertions of leaves, but deletions can be handled in a
  straightforward manner, as discussed in \autoref{section:preliminaries}.} of
leaves.
\subparagraph*{Subclasses of context-free languages.}  We tried to improve on the $\bigO(n^7)$ upper work bound for
  context-free languages, but did not succeed yet. The other goal of
  this paper is thus to find better bounds for important subclasses of
  the context-free languages: deterministic context-free languages and
  visibly pushdown languages.  We show that, for each $\epsilon>0$,
there are constant-time dynamic algorithms with work
$\bigO(n^{3+\epsilon})$ for deterministic context-free languages and
$\bigO(n^{2+\epsilon})$ for visibly pushdown languages. Here, the main
challenge is to carefully apply the technique from \cite{SchmidtSTVZ21} that
allows to store information for only \neps as opposed to $n$ different
values for some parameters. For more restricted change operations, the
algorithm for regular tree languages yields an  \neps work algorithm
for visibly pushdown languages.

\subparagraph*{Structure of the paper}
We explain the framework in \autoref{section:preliminaries}, and present the results on
regular tree languages and context-free string languages in
Sections~\ref{section:treelanguages} and \ref{section:contextfree},
respectively.  Almost all proofs are delegated to the \myappendix.

\subparagraph*{Related work}
In \cite{TschirbsVZ23}, parallel dynamic algorithms for regular string
languages under bulk changes were studied. It was shown that
membership in a regular language can be maintained, for every $\epsilon>0$, in constant time
with work \neps, even if a polylogarithmic number of changes can be
applied in one change operation. If the language is star-free,
polylogarithmic work suffices. The paper also shows that for regular
languages that are not star-free, polylogarithmic work does \emph{not}
suffice.

Maintaining regular languages of trees under label changes has also
been studied in the context of  enumeration algorithms (for
non-Boolean queries) \cite{AmarilliBM18}. The dynamic parallel
algorithms of  \cite{SchmidtSTVZ21}  partially rely on dynamic
sequential algorithms, especially \cite{FrandsenMiltersen+1997}.

\subparagraph*{Acknowledgements.} We are grateful to Jens Keppeler and
Christopher Spinrath for
careful proof reading.

\section{Preliminaries}\label{section:preliminaries}
\subparagraph*{Trees and regular tree languages.}
We consider ordered, unranked trees $t$, which we represent 
as tuples $(V,r,c,\text{label})$, where $V$ is a finite set of nodes,
$r\in V$ is the root, 
$c: V \times \N \to V$ is a function, such that $c(u,i)$ yields the
$i$-th child of $u$,  and $\text{label}: V \to \Sigma$ is a function that assigns a label to every node. 

We denote the set of unranked trees over an alphabet $\Sigma$ as $T(\Sigma)$. 
The terms \emph{subtree}, \emph{subforest}, \emph{sibling}, \emph{ancestor}, \emph{descendant}, \emph{depth} and $\emph{height}$ of nodes are defined as usual. A node that has no child is called a $\textit{leaf}$. A \emph{forest} is a sequence of trees.

Let $\preceq$ denote the order on siblings, i.e., $u\prec v$ denotes
that $u$ is a sibling to the left of $v$. We write $u\preceq v$ if
$u\prec v$ or $u=v$ holds.

By $\tree{v}$ we denote the subtree of $t$ induced by node $v$. For
sibling nodes $u\prec v$, we  write $\tree[u]{v}$ for the subforest of the tree $t$, induced by the sequence $u,\ldots,v$. 
If $w$ is a node in
$\tree{v}$, then $\tree{v}[w]$ denotes the subtree consisting
of $\tree{v}$ without $\tree{w}$. Analogously, for $\tree[u]{v}[w]$.

Our definition of tree automata is inspired from hedge automata in the
TaTa book \cite{Comon08}, slightly adapted for our needs.
 \begin{definition}
 	A \emph{deterministic finite (bottom-up) tree automaton (DTA)} over an alphabet $\Sigma$ is a tuple $\calB=(Q_\mathcal{B},\Sigma,Q_f,\delta,\calA)$ where $Q_\mathcal{B}$ is a finite set of states, $Q_f\subseteq Q_\mathcal{B}$ is a set of final states, $\calA=(Q_\calA,Q_\calB,\delta_\calA,s)$ is a DFA over alphabet $Q_\mathcal{B}$ (without final states) and $\delta: Q_\calA\times \Sigma \to Q_\mathcal{B}$ maps pairs $(p,\sigma)$, where $p$ is a state of $\calA$ and $\sigma\in\Sigma$, to states of $\calB$.  
      \end{definition}
      We refer to states from $Q_\mathcal{B}$ as \emph{$\calB$-states}
      and typically denote them by the letter $q$. Likewise states
      from $Q_\mathcal{A}$ are called \emph{$\calA$-states} and
      denoted by $p$. We note that we do not need a set of accepting states for $\calA$, since its final states are fed into $\delta$.  
      
      The semantics of DTAs is defined as follows. 
      
   For each tree $t\in T(\Sigma)$, there is a unique \emph{run} of $\calB$ on $t$, that is, a unary function $\rho_t$ that assigns a $\calB$-state to each node in $V$. It can be defined in a bottom-up fashion, as follows. 
      For each node $v\in V$ with label $\sigma$ and  children
      $u_1,\ldots,u_\ell$,  $\rho_t(v)$ is the $\calB$-state
      $\delta(\delta^*_\calA(s,\rho_t(u_1)\cdots
      \rho_t(u_\ell)),\sigma)$.  
      That is, the state of a node $v$ with label $\sigma$ is
      determined by $\delta(p,\sigma)$, where $p$ is the final $\calA$-state
      that $\calA$ assumes when reading the sequence of states of
      $v$'s children, starting from the initial state $s$. In
      particular, if $v$ is a leaf with label $\sigma$, its $\calB$-state is $\delta(s,\sigma)$.

 A tree $t$ is accepted by the DTA $\calB$ if $\rho_t(r)\in Q_f$ holds for the root $r$ of $t$. We denote the language of all trees accepted by $\calB$ as $L(\calB)$. We call the languages decided by DTAs \emph{regular}. 

 \subparagraph*{Strings and context-free languages.}

 Strings $w$ are finite sequences of symbols from an alphabet $\Sigma$. By $w[i]$ we denote the $i$-th symbol of $w$ and  by $w[i,j]$ we denote the substring from position $i$ to $j$.  We denote the empty string by $\lambda$, since $\epsilon$
    has a different purpose in this paper. We use standard notation
    for context-free languages and pushdown automata, to be found in
    the \myappendix.

 \subparagraph*{Dynamic algorithmic problems.}

In this paper, we view a dynamic (algorithmic) problem basically as
the interface of a data type: that is, there is a collection of
operations by which some object can be initialised, changed, and
queried. A \emph{dynamic algorithm} is then a collection of
algorithms, one for each operation. We consider two main dynamic
problems in this paper,  for regular tree languages and context-free languages.

For each regular tree language $L$, the algorithmic problem $\RegTree(L)$ maintains a labelled tree $T$ and
has the following 
operations.
\begin{itemize}
\item $\Init(T,r,\sigma)$  yields an initial labelled tree object $T$ and returns in $r$ a node id for its root, which is labelled by $\sigma$;
\item $\Relabel(T,u,\sigma)$ changes the label of node $u$ in $T$ into $\sigma$;
\item $\AddChild(T,u,v,\sigma)$ adds a new child with label $\sigma$ behind the last child of node $u$ and returns its id in $v$;
\item $\Query(T,v)$ returns true if and only if the subtree of $T$ rooted at $v$ is in $L$.
\end{itemize}
We refer to the restricted problem without the operation $\AddChild$
as \RegTreeM. For this data type, we assume that the computation starts
from an initial non-trivial
tree and that the auxiliary data for that tree is given initially.

For each context-free language $L$, the algorithmic problem $\CFL(L)$ maintains a string $w$ and
has the following 
operations.
\begin{itemize}
\item $\Init(w)$  yields an initial string object $w$ with an empty string;
\item $\Relabel(w,i,\sigma)$ changes the label at position $i$ of $w$ into $\sigma$;
\item $\InsertPositionBefore(w,i,\sigma)$ and $\InsertPositionAfter(w,i,\sigma)$ insert a new position with symbol $\sigma$ before or after the current position $i$, respectively;
\item $\Query(w,i,j)$ returns true if and only if the substring $w[i,j]$ is in $L$.
\end{itemize}

Readers may wonder, why these dynamic problems do not have operations that delete nodes of a tree or positions in a string. This is partially to keep the setting simple and partially because node labels and symbols offer easy ways to simulate deletion by extending the alphabet with a symbol $\blank$ that indicates an object that should be ignored. E.g., if $\delta_\calA(p,\blank)=p$, for every state $p$ of the horizontal DFA of a DTA, then the label $\blank$ at a node $u$ effectively deletes the whole subtree induced by $u$ for the purpose of membership in $L(\calB)$. Similarly, a CFL might have a neutral symbol or even  a pair $(_\blank,)_\blank$ of ``erasing'' brackets that make the PDA ignore the substring between  $(_\blank$ and $)_\blank$.

For $\RegTree(L)$ and $\CFL(L)$, the $\Init$
operation is possible in
constant sequential time and will not be considered in detail.

Throughout this paper, $n$ will denote an upper bound of the size of
the structure at hand (number of nodes of a tree or positions of a
string) that is linear in that size, but changes only infrequently. More precisely, the number
of nodes of a tree or the length of the string will always be between
$\frac{1}{4}n$ and $n$. Whenever the size of the structure grows
beyond $\frac{1}{2}n$, the data structure will be prepared for
structures of size up to $2n$ and, once this is done, $n$ will be
doubled. Since the size of the structure is always $\theta(n)$ all
bounds in $n$ also hold with respect to the size of the structure.

\subparagraph*{Parallel Random Access Machines (PRAMs).}
A \emph{parallel random access machine} (PRAM) consists of a number of
processors that work in parallel and use a shared
memory. The memory
is comprised of memory cells which can be accessed by a processor in
$\bigO(1)$ time.
Furthermore, we assume that simple arithmetic and bitwise
operations, including addition, can be done in $\bigO(1)$ time by a processor.
We mostly use
the Concurrent-Read Concurrent-Write model (CRCW PRAM),
i.e. processors are allowed to read and write concurrently from and to
the same memory location. More precisely, we assume the
\emph{common} PRAM model: several processors can concurrently write into the same memory location,  only if
  all of them write the same value. We also mention the
  Exclusive-Read Exclusive-Write model (EREW PRAM), where concurrent
  access is not allowed. 
The work  of a
PRAM computation is the sum of the number of all computation steps of
all processors made during the computation.
We define the space $s$ required by a PRAM computation as the maximal index of any memory cell accessed during the computation.
We refer to \cite{DBLP:books/aw/JaJa92} for more details on PRAMs and
to \cite[Section 2.2.3]{DBLP:books/el/leeuwen90/Boas90} for a
discussion of alternative space measures.

The main feature of the common CRCW model 
  relevant for our algorithms that separates it from the EREW model is that it allows to compute the minimum
  or maximum value of an array of size $n$ in constant time (with work $\bigO(n^{1+\epsilon})$) which is shown in another paper at MFCS 2023.\footnote{Jonas Schmidt, Thomas Schwentick. Dynamic constant time parallel graph algorithms with sub-linear work.}

For simplicity, we assume that even if the size bound $n$ grows,
a number in the range $[0,n]$ can still be stored in one memory
cell. This assumption is justified, since addition of larger numbers $N$
can still be done in constant time and polylogarithmic work on a CRCW
PRAM.
Additionally, we assume that the number of processors  always depends
on the current size bound $n$. Hence, the number of processors
increases with growing $n$ which allows us to use the PRAM model
with growing structures.

We describe our PRAM algorithms on an abstract level and do not
exactly specify how processors are assigned to data. Whenever an
algorithm does something in parallel for a set of objects, these
objects can be assigned to a bunch of processors with the help of some
underlying array. This is relatively straightforward for strings and
substrings and the data structures used in
\autoref{section:contextfree}. In \autoref{section:treelanguages}, it
is usually based on zone records and their underlying partition records.

\section{Maintaining regular tree languages}\label{section:treelanguages}
In this section, we present our results on maintaining regular tree languages under various change operations.
We will first consider only operations that change node labels, but do not change the shape of the given tree.  A very simple dynamic algorithm with work $\bigO(n^2)$ is presented in the \myappendix. We sketch its main idea and how it can be improved to $\neps$ work per change operation by using a \emph{partition hierarchy} in \autoref{subsection:labelchanges-efficient}. These algorithms even work on the EREW PRAM model.

Afterwards, in \autoref{subsection:structuralchanges}, we also consider an operation that can change the tree structure: adding a leaf to the tree. Here, the challenge is to maintain the hierarchical structure that we used before to achieve work $\neps$ per change operation. It turns out that maintaining this structure is possible without a significant increase of work, that is, maintaining membership under these additional operations is still possible with work $\neps$ per change operation.

\subsection{Label changes: a work-efficient dynamic program}\label{subsection:labelchanges-efficient}

In this section, we describe how membership in a regular tree language
can be maintained under label changes, in a work
efficient way.
             \begin{proposition}\label{prop-label-changes}
                    For each $\epsilon>0$ and each regular tree language $L$,    there is a parallel constant time dynamic algorithm
                for $\RegTreeM(L)$  with work $\bigO(n^\epsilon)$ on an EREW PRAM. The $\Query$ operation can actually be answered with constant work.
              \end{proposition}
          
 We start by briefly sketching the $\bigO(n^2)$ work algorithm that is given in the \myappendix.
 The algorithm basically combines the dynamic programs for regular string languages and binary regular tree languages from \cite{GeladeMS12}.
 For regular string languages, the program from \cite{GeladeMS12} stores the behaviour of a DFA for the input word $w$ by maintaining information of the form "if the run of the DFA starts at position $i$ of $w$ and state $p$, then it reaches state $q$ at position $j$" for all states $p,q$ and substrings $w[i,j]$. After a label change at a position $\ell$, this information can be constructed by combining the behaviour of the DFA on the intervals $w[i,\ell-1]$ and $w[\ell+1, j]$ with the transitions induced by the new label at position $\ell$.
 
 The dynamic program for (binary) regular tree languages from \cite{GeladeMS12} follows a similar idea and stores the behaviour of a (binary) bottom-up tree automaton by maintaining information of the form "if $v$ gets assigned state $q$, then $u$ gets assigned state $p$ by the tree automaton" for all states $p,q$ and all nodes $v,u$, where $v$ is a descendant of $u$.
 
 Both programs induce algorithms with $\bigO(n^2)$ work bounds. Towards a $\bigO(n^2)$ work algorithm for  unranked tree languages, the two dynamic programs can be combined into an algorithm that mainly stores the following $\emph{automata functions}$ for  a fixed DTA $\mathcal{B}=(Q_\calB,\Sigma, Q_f, \delta,
\calA)$ for $L$, with  DFA $\calA=(Q_\calA,Q_\calB,\delta_\calA,s)$:
\begin{itemize}
	\item  The ternary function $\calB_t: Q_\calB \times V\times V\mapsto Q_\calB$ maps
	each triple $(q,u,v)$ of a state $q \in Q_\calB$ and nodes of $t$, where $u$ is a proper ancestor of
	$v$, to the state that the run of $\calB$ on $\tree{u}[v]$  takes at
	$u$, with the provision that the state at $v$ is $q$.
	\item  The ternary function $\calA_t: Q_\calA \times V\times V\mapsto Q_\calA$ maps
	each triple $(p,u,v)$ of a state $p \in Q_\calA$ and nodes of $t$, where $u\prec v$ are siblings, to
	the state that the run of $\calA$ on $u,\ldots,v$, starting from
	state $p$,  takes after $v$.
\end{itemize}

Every single function value can be updated in constant sequential time, as stated in the following lemma. This leads to a quadratic work bound since there are quadratically many tuples to be updated in parallel.

\begin{restatable}{lemma}{LemmaUpdateTuple} \label{lemma-update-tuple}
	After a $\Relabel$ operation, single values $\calA_t(p,u,x)$ 
	and $\calB_t(q,u,x)$ can be updated  by a sequential algorithm
	in constant time.
\end{restatable}

Some information about the shape of the tree is required, which we refer to as \emph{basic tree functions}. For more details we refer to the \myappendix. However, as label changes cannot change the shape of the tree, this information does not need to be updated und can be assumed as precomputed.

To lower the work bound the basic idea now is to store the automata functions not for \emph{all} possible arguments, but for a small subset of \emph{special} arguments that allow the computation of function values for \emph{arbitrary} arguments in constant time with constant work.

In \cite{SchmidtSTVZ21}, this idea was applied to the $\bigO(n^2)$ work program for regular string languages. A constant-depth hierarchy of intervals was defined by repeatedly partitioning intervals into $\ntheta$ subintervals, for some $\theta>0$. This hierarchy allowed to define \emph{special} intervals such that any update only affects $\neps$ intervals and function values of arbitrary intervals can be computed in constant time with constant work.

We transfer this idea to the case of unranked tree languages by partitioning the tree into \ntheta \emph{zones}, each of which is partitioned into further \ntheta zones and so on until, after a constant number of refinements, we arrive at zones of size \ntheta. Here, $\theta>0$ is a constant that will be  chosen later. It will always be chosen such that $h=\frac{1}{\theta}$ is an integer.

Before we define this partition hierarchy more precisely, we first define zones and show that they can always be partitioned in a way that guarantees certain number and size constraints.

\begin{defi}
	A \emph{zone} is a set $S$ of nodes with the following properties:
	\begin{itemize}
		\item $S$ is a proper subforest of $t$,
		\item for every $v\in S$ it holds that either no or all children are in S, and
		\item there exists at most one node $v_S$ in $S$, whose children are not in $S$. The node $v_S$ is called the \emph{vertical connection node} of $S$.
	\end{itemize}

      \end{defi}
      We call a zone a \emph{tree zone} if it consists of only one sub-tree of $t$ and a \emph{non-tree zone} otherwise. We call a zone \emph{incomplete} if it has a vertical connection node and \emph{complete}, otherwise. There are thus four different types of zones which can be written, with the notation introduced in \autoref{section:preliminaries}, as follows: complete tree zones $\tree{v}$, complete non-tree zones $\tree[u]{v}$, incomplete tree zones $\tree{v}[w]$, and incomplete non-tree zones $\tree[u]{v}[w]$. Depending on the type, zones can therefore be represented by one to three ``important nodes''. The overall tree can be seen as the zone $\tree{r}$, where $r$ is its root.

From now on, we always assume that $n$ is as in \autoref{section:preliminaries}, some $\theta>0$ is fixed, and that $h=\frac{1}{\theta}$ is an integer. 
      
We call a zone of $t$ with at most $n^{\theta\ell}$ nodes an \emph{$\ell$-zone}. The tree $t$ itself constitutes a $h$-zone, to which we will refer to as the \emph{overall zone}. 

We next define \emph{partition hierarchies} formally. More precisely, for every $\ell\ge 2$, we define partition hierarchies of height $\ell$ for $\ell$-zones as follows. 
If $S$ is a 2-zone and $S_1,\ldots,S_k$ are 1-zones that constitute a partition of $S$, then $(S,\{S_1,\ldots,S_k\})$ is a partition hierarchy of height 2 for $S$.
If $S$ is an $(\ell+1)$-zone, $\{S_1,\ldots,S_k\}$ is a partition of $S$ into $\ell$-zones, and for each $j$, $H_j$ is a partition hierarchy of height $i$ for $S_j$, then $(S,\{H_1,\ldots,H_k\})$ is a partition hierarchy of height $\ell+1$ for $S$.
A partition hierarchy of height $h$ of the zone consisting of $t$ is called a partition hierarchy of $t$.

An example of a $(1,\frac{1}{3})$-bounded partition hierarchy is given in \autoref{fig:partition-hierarchy}.

\begin{figure}[t]
    \begin{center}
        \includegraphics[scale=0.155]{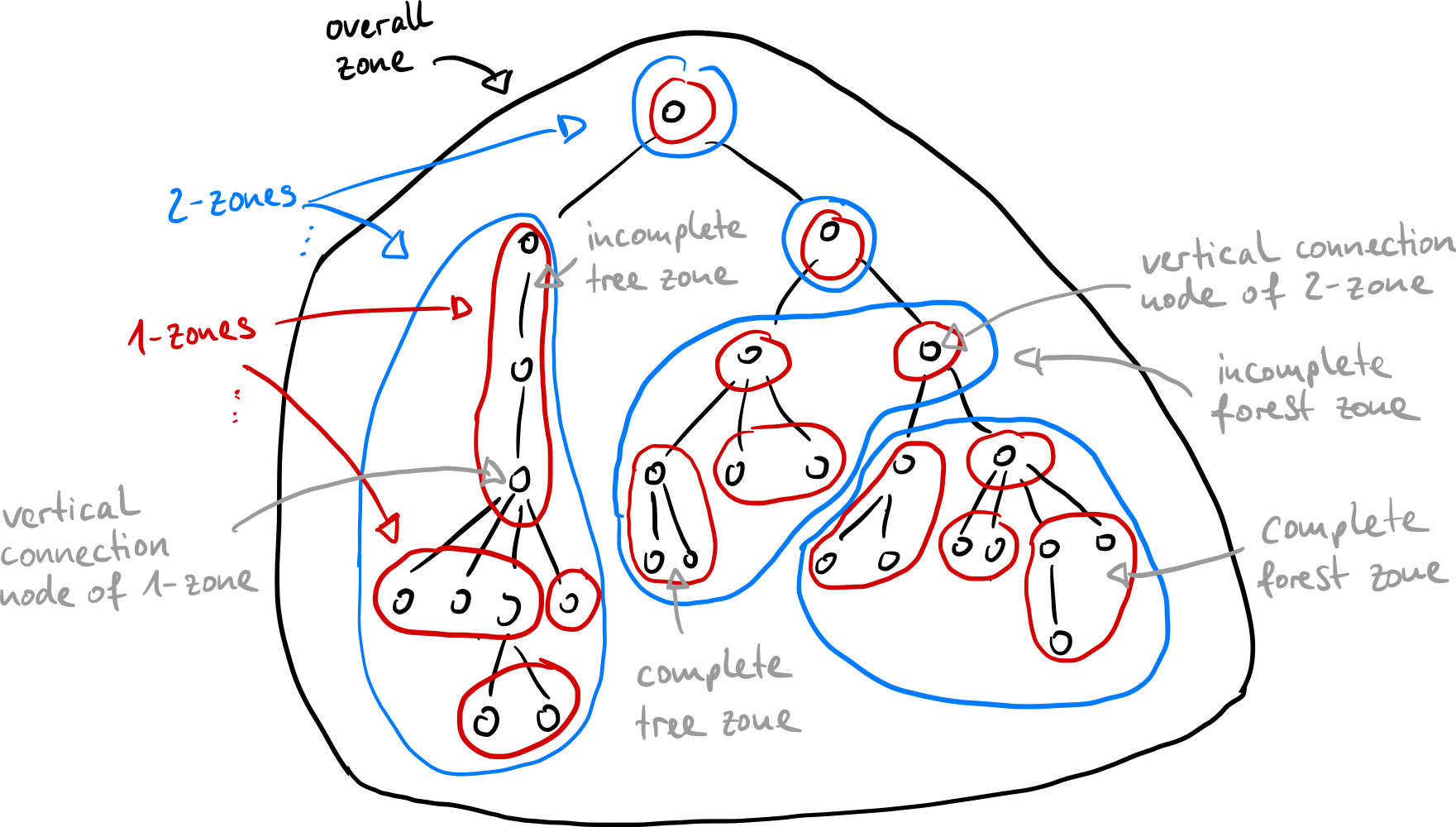}
    \end{center}
    \caption{Example of a $(1,\frac{1}{3})$-bounded partition hierarchy.}
    \label{fig:partition-hierarchy}
\end{figure}

We often call a zone $S'$ that occurs at some level $i<\ell$ within the partition hierarchy of a zone $S$ of some level $\ell$ a \emph{component zone}. If $S'$ has level $\ell-1$ we also call it a \emph{sub-zone} of $S$.  

We call a partition hierarchy $H$ \emph{$(c,\theta)$-bounded},  constants $c$ and $\theta>0$, if each partition of a zone consists of at most $cn^\theta$ nodes. 

Our next aim is to prove that $(10,\theta)$-bounded partition hierarchies actually exist. To this end, we prove the following lemma. It is similar to \cite[Lemma 3]{Bojanczyk12}, but adapted to our context, which requires a hierarchy of constant depth and a certain homogeneity regarding children of vertical connection nodes. 

\begin{restatable}{lemma}{LemmaPartitionStep}
	\label{lemma:pre-partition-step}
Let $m\ge 2$ be a number and $S$  a zone with more than $m$ nodes. Then $S$ can be partitioned into at most five zones, one of which has at least $\frac{1}{2}m$ and at most $m$ nodes.
\end{restatable}

      This lemma immediately yields the existence of $(10,\theta)$-bounded partition hierarchies.
      \begin{restatable}{proposition}{PropPartitionHierarchy}
     \label{proposition-partition-hierarchy}
        For each $\theta>0$, each tree $t$ has some $(10,\theta)$-bounded partition hierarchy.
      \end{restatable}

              We now explain in more detail, which information about the behaviour of $\calA$ and $\calB$ is stored by the work-efficient algorithm.

Function values for the ternary functions are stored only for so-called special pairs of nodes, which we define next. Special pairs of nodes are always defined in the context of  some zone $S$ of a partition hierarchy. In the following, we denote, for a zone $S$ of a level $\ell\ge 2$ its set of sub-zones of level $\ell-1$ by $T$.

\begin{itemize}
\item Any  pair of siblings $u\prec v$ in a zone $S$ of level 1 is a \emph{special horizontal pair}. A pair of siblings $u\prec v$ in a complete zone $S$ of level $\ell\ge 2$ is a \emph{special horizontal pair}, if  $u$ is a left boundary of some zone in $T$ and $v$ is a right boundary  of some zone in $T$. However, if $S$ is incomplete and  there is an ancestor $w'$ of the lower boundary $w$ with $u\preceq w' \preceq v$, then, instead of $(u,v)$,   there are two special pairs: $(u,\lsibling(w'))$ and $(\rsibling(w'),v)$.

\item  Any  pair of nodes $u,v$ in some zone $S$ of level 1 is a \emph{special vertical pair}, if  $v$ is an ancestor of $u$. A pair of nodes $u,v$ in some zone $S$ of level $\ell\ge 2$ is a \emph{special vertical pair}, if  $v$ is an ancestor of $u$, $v$ is an upper or lower boundary of some zone in $T$ and $u$ is a lower boundary of some zone in $T$. However, if $S$ is incomplete with lower boundary $w$ and  $w':=\lca(w,u)$ is strictly above $u$ and below or equal to $v$, then, instead of $(u,v)$, there are two special pairs: $(u,\ancchild(w',u))$ and $(w',v)$. Here \lca\ determines the least common ancestor and \ancchild\ the child of $w'$ that is an ancestor of $u$.
\end{itemize}

The algorithm stores $\calA_t(p,u,v)$ for each state $p$ of $\calA$ and each special horizontal pair $u,v$.  Furthermore, it stores $\calB_t(q,u,v)$, for each state $q$ of $\calB$ and each special vertical pair $u,v$.

We note, that in all cases  $\calA_t(p,u,v)$ and $\calB_t(q,u,v)$ only depend on the labels of the nodes in the zone, for which $(u,v)$ is special. 

\begin{restatable}{lemma}{LemmaComputeFromStored}
              \label{lemma-compute-from-stored}
                From the stored values for functions $\calA_t$ and $\calB_t$ for special pairs, it is possible to compute $\rho_t(v)$, for arbitrary nodes $v$, $\calA_t(p,u,u')$ for arbitrary pairs $u\prec u'$ of siblings of $t$ and $\calB_t(q,u,u')$ for arbitrary pairs $u,u'$ of nodes, where $u'$ is an ancestor of $u$, sequentially in constant time.
              \end{restatable}
          
This enables us to show the \neps\ work bound for label changes.

 \begin{proof}[Proof of \autoref{prop-label-changes}]
	To achieve the stated bound, we use the above algorithm with work parameter $\theta=\frac{\epsilon}{2}$. The algorithm uses a $(\theta,10)$-bounded partition hierarchy, which exists thanks to \autoref{proposition-partition-hierarchy}.
	
	As indicated before, the algorithm stores $\calA_t(\cdot,u,v)$ and $\calB_t(\cdot,u,v)$, for all special pairs $(u,v)$. As already observed before, these values only depend on the labels of the nodes of the zone relative to which $(u,v)$ is special. Therefore, if a node label is changed for some node $x$, values   $\calA_t(\cdot,u,v)$ and $\calB_t(\cdot,u,v)$ need only be updated for special pairs of zones in which $x$ occurs. Since each node occurs in exactly $h$ zones and each zone has $\bigO(n^{2\theta})=\neps$ special pairs, $h\cdot\neps$ processors can be used, where every processor updates a single value in constant time and work, as is possible thanks to \autoref{lemma-update-tuple} and \autoref{lemma-compute-from-stored}. Since the shape of the tree does not change we can assume a mapping from the updated node and the processor number to the special tuple that the respective processor recomputes.
\end{proof}

\subsection{Structural Changes}\label{subsection:structuralchanges}

In \autoref{prop-label-changes} only label changes were allowed, so the structure of the underlying tree did not change. In particular, there was no need to update any of the basic tree functions.

In this subsection, we consider structural changes of the tree. We show that the work bounds of \autoref{prop-label-changes} can still be met for the full data types $\RegTree(L)$.

 \begin{theorem}\label{theorem-struct-changes}
	For each regular tree language $L$ and each $\epsilon>0$, there is a dynamic constant time parallel algorithm for $\RegTree(L)$ that handles change operations with work $\bigO(n^\epsilon)$  and answers query operations with constant work.
      \end{theorem}

      In the next subsection, we describe the general strategy of the algorithm, define some notions that will be used and present its proof.
      Then, in a second subsection, we give some more detailed information about the data that is stored and how it can be maintained.

\subsubsection{High-level description of the dynamic algorithm}\label{subsubsection:structuralchanges-high-level}

Our approach generalises the algorithm of  \autoref{subsection:labelchanges-efficient}.  It makes sure that, at any point in time, there is a valid partition hierarchy together with corresponding tree and automata functions.
 The general strategy of the dynamic algorithm is to add new leaves  to their  nearest zone. In principle, this is not hard to handle --- unless it leads to a violation of a size constraint of some zone.  As soon as zones exceed a certain size bound the affected parts of the hierarchy will thus be  recomputed to ensure the size constraints.

For reasons that will become clearer below, we need to slightly modify the definition of partition hierarchies, basically by omitting the lowest two levels. To this end, we define $3$-pruned partition hierarchies just like we defined partition hierarchies, but the lowest level is at height $3$. More precisely, \emph{a $3$-pruned partition hierarchy of height $3$} is just a $3$-zone, and  \emph{$3$-pruned partition hierarchies of height $\ell>3$} are inductively defined just like partition hierarchies of height $\ell$.
It is clear that a $3$-pruned partition hierarchy exists for each tree by ommiting the two lowest levels in the partition hierarchy computed in \autoref{proposition-partition-hierarchy}. Moreover, using a $3$-pruned partition hierarchy as basis for our efficient label change approach still ensures the sequential constant time computation of arbitrary automaton function values from the stored values for special pairs. However, zones on the lowest level have size $\bigO(n^{3\theta})$ leading to a work bound of $\bigO(n^{6\theta})$ per change operation.

To ensure that at each point in time, a usable partition hierarchy is available, the general strategy is as follows: the algorithm starts from a \emph{strong partition hierarchy} in which zones at level $\ell$ have size at most $\frac{1}{4}n^{\ell\theta}$, well below the maximum allowed size of such a zone of $n^{\ell\theta}$. As soon as the size of a zone $S$ at level $\ell$ reaches its \emph{warning limit} $\frac{1}{2}n^{\ell\theta}$, the algorithm starts to compute a new partition hierarchy for the parent zone $S'$ of $S$ at level $\ell+1$. This computation is orchestrated in a way that makes sure that the new partition hierarchy for $S'$ is ready (together with all required function values) before $S$ reaches its size limit $n^{\ell\theta}$, at which point the old partition hierarchy for $S'$ becomes useless.

Since a partition hierarchy of the whole tree together with the required function values  has size  $\Omega(n)$, its computation inherently requires that amount of work and it can probably not be done in constant time.  Furthermore, since we aim at work \neps per  operation, the algorithm cannot afford to do the re-computation ``as fast as possible'' but rather needs to stretch  over at least $n^{1-\epsilon}$ steps. However, the fact that the tree can change during a re-computation poses a challenge: if many change operations happen with respect to  a particular zone in a low level of the new partition hierarchy, this new zone might reach its warning limit and then its hard limit, before the overall re-computation of the hierarchy has finished. This challenge can be met by a careful orchestration of the re-computation.

We will next describe the data structure that the dynamic algorithm uses to orchestrate re-computations of partition hierarchies. As mentioned before, there will always be a valid partition hierarchy. However, for some zones, re-computations might be underway. The algorithm will always manage to complete the re-computation of a partition hierarchy for a zone of level $\ell$, before any of the subzones of level $(\ell-1)$ of the new partition reaches its warning limit. Therefore, for each zone within the data structure, there is always at most one partition hierarchy under construction, and therefore each zone has at any time at most two partition records.
If a zone actually has two partition records, one of them contains a usable partition hierarchy. We formalise usability of a partition hierarchy by the term \emph{operable} and tie the whole data structure together  through the following notion of zone records. It is defined in an inductive fashion, together with the concept of partition records.

\begin{defi}
	A \emph{zone record} of level $3$ is a 3-zone.
	A \emph{zone record} of level $\ell>3$ consists of an $\ell$-zone $S$ and up to two partition records $P_1,P_2$ of level $\ell$ for $S$. If it has two partition records then $P_1$ is complete and $P_2$ is incomplete.
	
	A \emph{partition record} $(Z,M)$ of level $\ell>3$ for an $\ell$-zone $S$ consists of a set $Z$ of zone records of level $\ell-1$ and a set $M$ of zones, such that the zones from $Z$ and the zones from $M$ together constitute a partition of $S$.
	 A partition record $Z$  of level $\ell$ is \emph{valid}, if all zones of its zone records are actual $(\ell-1)$-zones.
	
	A zone record of level $3$ is  \emph{operable}.
	
	A partition record at level $\ell>3$ is \emph{operable}, if it is valid and all its zone records are operable.
    A zone record of level $\ell>3$ is \emph{operable}, if its first partition record is operable. 
  \end{defi}
  We refer to the hierarchical structure constituting the overall zone record as the \emph{extended partition hierarchy}.
Within the extended partition hierarchy, we are particularly interested in ``operable substructures''.   
To this end, we associate  with an operable zone record, the \emph{primary partition hierarchy} that results from recursively picking the operable partition record from each zone record. 

Altogether, the algorithm maintains an extended partition hierarchy for $t$.

Before we describe how the algorithm stores the extended partition hierarchy, we need two more concepts.
For each zone record $R$ of a level $\ell$ there is a sequence $R_h,\ldots,R_\ell=R$ of zone records such that, for each $i\ge \ell$, $R_i$ is a zone record that occurs in a partition record of $R_{i+1}$. This sequence can be viewed as the \emph{address} of $R$ in the extended partition hierarchy. Furthermore, this address induces a \emph{finger print} for $R$: the sequence $\status(R_h),\ldots,\status(R_\ell)$, where $\status(R_i)$ is either \emph{operable} or \emph{in progress}. It is a simple but useful observation that if a tree node $v$ occurs in two zones with zone records $R\not=R'$ within the extended partition hierarchy, then the finger prints of $R$ and $R'$ are different. Consequently a tree node occurs in at most $2^h$ and thus, a constant number of zones in the extended partition hierarchy.

Now we can describe, how the algorithm stores $t$ and the extended partition hierarchy. 
\begin{itemize}
	\item A zone record of level $3$ is represented as an array of \ntheta nodes.
	\item A zone record of a level $\ell>3$ consists of up to four boundary nodes and up to two pointers to partition records. The operable partition record is flagged.
        \item Each  zone record of level $\ell\geq 3$ with finger print $pa$, also stores a pointer to its zone on level $\ell+1$ with finger print $p$, and three pointers to the zone records of its parent, first child and right sibling zones. 
	\item A partition record $(Z,M)$ is represented as an array of zone records (some of which may be zones of $M$). The zones records from $Z$ are flagged.
        \item The nodes of $t$ are stored in an array (in no particular order) together with pointers for the functions $\parent$, $\lsibling$, $\rsibling$, $\fchild$, and $\lchild$.
	\item For each node $v$, and each possible finger print $p$,   a pointer $Z^p(v)$ to its zone with finger print $p$.
\end{itemize}

Now we are prepared to outline the proof of \autoref{theorem-struct-changes}.

\begin{proofof}{\autoref{theorem-struct-changes}}
  Let $\calB$ be a DTA for the regular tree language $L$ and let
  $\theta=\frac{\epsilon}{7}$.
  The dynamic algorithm stores $t$ and an extended partition hierarchy as described above. It also stores some additional function values, including values for the automata functions, that will be specified in  \autoref{subsubsec:treefunctions}.

  Some functions are independent from zones and are stored for all nodes. Some other functions are  independent from zones but are only stored for particular node tuples  that are induced from zones (like it was already the case for the automata functions in \autoref{subsection:labelchanges-efficient}) and some functions are actually defined for (tuples of) zones. 
  
  After each change operation, the algorithm updates function values, pursues re-computations of hierarchies and computes function values that are needed for newly established zones. It starts a re-computation for a zone $S$, whenever one of its subzones reaches its warning limit. It starts a re-computation of the overall zone, whenever the number of nodes of $t$ reaches $\frac{1}{2}n$.

  The algorithm has one thread for each zone with an ongoing re-computation, that is, for each zone whose zone record is not yet operable.

  A re-computation for a zone at level $\ell$ requires the computation of $\ntheta$ zones of level $\ell-1$, each of which yields re-computations of \ntheta zones of   level $\ell-2$ and so forth, down to level 3. It is easy to see that the overall number of zones that needs to be computed during a re-computation of a zone at level $\ell$ is bounded by $\ntheta[(\ell-3)]$. The re-computation of the overall zone requires the computation of at most $\ntheta[1-3]$ zones. We show in \autoref{lemma-partition-efficient} that, in the presence of a primary partition hierarchy for the overall zone, the computation of a new zone is possible in constant time with work $\ntheta[6]$.

  The thread for the re-computation of a zone at level $\ell$ thus (first) consists of $\ntheta[(\ell-3)]$ computations of  component zones, each of which is carried out in constant time with work  $\ntheta[6]$. We refer to such a re-computation as a round. A thread thus consists of $\ntheta[(\ell-3)]$ rounds of zone computations. The thread follows a breadth-first strategy, that is, it first computes all zones of level $\ell-1$ then the sub-zones of those zones at level $\ell-2$ and so forth. Once the zone record of a zone $S$ is operable, the thread computes in its second phase  all function values associated to $S$. 
  This can be done in constant time with work $\ntheta[7]$ per sub-zone of $S$, as is shown in the \myappendix. That is, it requires at most $\ntheta[(\ell-3)]$ additional rounds. 

  We note that it does not matter if the primary partition hierarchy $H$ required for \autoref{lemma-partition-efficient} changes during the computation of a thread, since $H$ is only used to make the identification of a new zone more efficient.
   
  To address the above mentioned challenge, the algorithm starts a separate thread for each zone that is newly created during this process. That is, for each zone at level $\ell-1$, an additional re-computation thread is started, as soon as the zone is created.

   Now we can state the orchestration strategy for re-computations. This strategy is actually very simple:
   \begin{quotation}
   \textbf{Re-computation strategy:} 
After each change operation affecting some node $v$, the algorithm performs one computation round, for all threads of zones $S$, at any level,  with $v\in S$.
\end{quotation}
That is, thanks to the above observation, after a change operation, there  are at most $2^h$ threads for which one computation round is performed. Since $2^h$ is a constant, these computations together require work at most $\ntheta[7]$.

On the other hand, the whole re-computation for a zone $S$ at level $\ell$, including the computation of the relevant function values, is finished after at most $\ntheta[(\ell-3)]$ change operations that affect $S$. Since $\frac{1}{2}n^{(\ell-1)\theta}$ leaf additions are needed to let a sub-zone $S'$ grow from the warning limit  $\frac{1}{2}n^{(\ell-1)\theta}$ to the hard limit $n^{(\ell-1)\theta}$, it is guaranteed that the re-computation thread for $S$ is completed, before $S'$ grows too large. In fact, this is exactly, why  partition hierarchies are 3-pruned. When a re-computation of the overall zone was triggered by the size of $t$, $n$ is doubled as soon as this re-computation is completed.

Thanks to \autoref{lemma-compute-tree-functions} the overall work to update the stored function values for all affected zones (in constant time) after a change operation is $\ntheta[3]$.

Altogether, the statement of the theorem follows by choosing $\theta=\frac{\epsilon}{7}$. 
      \end{proofof}

      We state the lemma about the computation of new zones next. The partition hierarchy is used as a means to assign evenly distributed nodes to processors and to do parallel search for nodes with a particular property  regarding the number of their descendants.

\begin{restatable}{lemma}{LemmaPartitionEfficient}
\label{lemma-partition-efficient}
 Given a tree $t$, a $3$-pruned partition hierarchy $H$ of $t$, and a zone $S$ with at least $m$ nodes, $S$ can be partitioned into at most five zones, one of which has at least $\frac{1}{2}m$ and at most $m$ nodes, in constant time with work $\bigO(n^{6\theta})$.
\end{restatable}

\subsubsection{Maintaining functions}\label{subsubsec:treefunctions}

In \autoref{subsection:labelchanges-efficient},  the tree functions were static and given by the initialisation. Only the automata functions needed to be updated. However, if leaf insertions are allowed, the tree functions can change. To keep the algorithm efficient, the  special pairs need to be adapted to the evolution of the partition hierarchy, and tree functions can no longer be stored for all possible arguments. Furthermore, additional tree functions and functions defined for zones will be used.

The stored information suffices to compute all required functions in constant time, and almost all of them with constant work.

\begin{restatable}{lemma}{LemmaTreeFunctions}
\label{lemma-compute-from-stored-tree-functions}
Given a tree $t$, a $3$-pruned partition hierarchy $H$ of $t$, and the stored information as described above, for each $\theta>0$, the $\child$ function can be evaluated in $\bigO(1)$ time with work $\ntheta$. All other functions  can be evaluated for all tuples with constant work.
\end{restatable}
      
Furthermore,  all stored information can be efficiently updated, with the help of and in accordance with the current primary partition hierarchy.

\begin{restatable}{lemma}{LemmaUpdateTreeFunctions}
\label{lemma-compute-tree-functions}
	Let $\theta>0$ and $H$ be a $3$-pruned partition hierarchy of $t$ with automata and tree functions. The stored information described above can be maintained after each $\Relabel$ and $\AddChild$ operation in constant time with $\ntheta[6]$ work per operation.
\end{restatable}

\section{Maintaining context-free languages}\label{section:contextfree}
As mentioned in the introduction, an analysis  of the dynamic program that was used in \cite{GeladeMS12}
to show that context-free languages can be maintained in \DynFO
yields the following result.
\begin{theorem}[{\cite[Proposition 5.3]{GeladeMS12}}]\label{theo:cfl}
   For each   context-free language $L$, there is a dynamic constant-time parallel
   algorithm on a
    CRCW PRAM for $\CFL(L)$ with $\bigO(n^7)$ work.
  \end{theorem}
  There is a huge gap between this upper bound
and the conditional lower bound of
$\bigO(n^{\omega-1-\epsilon})$, for any $\epsilon>0$, derived from  the
$k$-Clique conjecture \cite{AbboudBackurs+2018}, where
$\omega<2.373$ \cite{SchmidtSTVZ21}.
Our attempts to make this gap significantly smaller, have not
been successful yet. However, for realtime deterministic context-free
languages and visibly pushdown languages, more efficient dynamic
algorithms are possible, as shown in the following two
subsections. 

\subsection{Deterministic context-free languages}\label{subsec:dcfl}
 Realtime deterministic context-free languages are decided by deterministic PDAs without $\lambda$-transitions (RDPDAs).

 \begin{restatable}{theorem}{TheoDCFL}\label{theo:dcfl}
   For each realtime deterministic context-free language $L$ and
   each $\epsilon>0$, there is a dynamic constant-time parallel
   algorithm on a
    CRCW PRAM for $\CFL(L)$ with $\bigO(n^{3+\epsilon})$ work.
  \end{restatable}
  Given an RDPDA $\calA$ for $L$,
  a configuration $C=(p,u,s)$ consists of a state $p$, a string $u$ that
is supposed to be read by $\calA$ and a string $s$, the initial stack
content.
We use the following functions $\dstate$, $\dstack$, and $\dempty$ to describe the
behaviour of $\calA$ on configurations.
\begin{itemize}
    \item $\dstate(C)$ yields the last state of $\run(C)$.
    \item $\dstack(C)$ yields the stack content at the end of  $\run(C)$.
         \item $\dempty(C)$ is the position in $u$, after which $\run(C)$ empties its stack. It is zero, if this does not
           happen at all.
         \end{itemize}
The algorithm maintains the following information, for each
simple configuration $C=(p,u,\tau)$, where $u=w[i,j]$, for some $i\le
j$, for
each suffix $v=w[m,n]$ of $w$, where $j<m$, each state $q$, and some $k\le n$.
\begin{itemize}
\item $\dall(C)$ defined as the tuple $(\dstate(C), |\dstack(C)|,
  \topk[1](\dstack(C)) ,  \dempty(C),)$, consisting of  the
  state,   the height of the stack, the top symbol of the stack,  
 at the end of the run on
  $C$ and the position where the run ends. If the run empties the
  stack prematurely or at the end of $u$, then $\topk[1](\dstack(C))$
  is undefined;
\item $\pushpos(C,k)$, defined as the length of the longest prefix
  $x$ of $u$, such that
  $|\dstack(p,x,\tau)|=k$. Informally this is the
  position of $u$ at which the $k$-th symbol of $\dstack(C)$, counted
  from the bottom,  is
  written;
  \item $\poppos(C,q,v,k)$, defined as the pair $(o,r)$, where $o$ is
    the length of the prefix
  $v'$ of $v$, for which $\run(q,v, \topk(\dstack(C)))$ empties its stack at
        the last symbol of $v'$,  and $r$ is the state it enters.
\end{itemize}
However, tuples for \pushpos and \poppos are only stored for values $k$ of the form
$an^{b\theta}$, for integers $b<\frac{1}{\theta}$ and $a\le n^\theta$,
for some fixed $\theta>0$.
A more detailed account is given in the \myappendix.

\subsection{Visibly pushdown languages}\label{subsec:vpl}

Visibly pushdown languages are a subclass of realtime deterministic
CFLs. They use \emph{pushdown alphabets} of the form
$\tilde{\Sigma}=(\Sigma_c,\Sigma_r, \Sigma_{\text{int}})$ and
deterministic PDA that always push a symbol when reading a symbol from
$\Sigma_c$, pop a symbol when reading a symbol from
$\Sigma_r$ and leave the stack unchanged otherwise. We refer to
\cite{AlurM04} for more information. 

There is a correspondence between wellformed strings over a pushdown
alphabet and labelled trees, where each matching  pair $(a,b)$ of a call symbol from
$\Sigma_c$ and a return symbol from $\Sigma_r$ is represented by an
inner node with label $(a,b)$ and each other symbol by a leaf. 
From \autoref{theorem-struct-changes} and this correspondence the
following can be concluded.

\begin{proposition}
   For each visibly pushdown language $L$ and
   each $\epsilon>0$, there is a dynamic constant-time parallel
   algorithm on a
    CRCW PRAM for $\VPL^-(L)$  with $\bigO(n^{\epsilon})$ work.
\end{proposition}
Here, $\VPL^-(L)$ only allows the following change operations:
\begin{itemize}
\item Replacement of a symbol by a symbol of the same type;
\item Insertion of an internal symbol from $\Sigma_{\text{int}}$
  before a return symbol;
\item Replacement of an internal symbol by two symbols $ab$, where
  $a\in\Sigma_c$ and $b\in\Sigma_r$.
\end{itemize}

For arbitrary symbol replacements and insertions, there is a much less
work-efficient algorithm which, however, is still considerably more
efficient than the algorithm for DCFLs. 

\begin{restatable}{theorem}{TheoVPL}\label{theo:vpl}
  For each visibly pushdown language $L$ and
   each $\epsilon>0$, there is a dynamic constant-time parallel
   algorithm on a
    CRCW PRAM for $\VPL(L)$ with $\bigO(n^{2+\epsilon})$ work.
  \end{restatable}
  The work improvement mainly relies on the fact that how the height
  of the stack evolves during a computation only depends on the types
  of symbols.

\section{Conclusion}\label{section:conclusion}
We have shown that the good work bounds for regular string languages
from \cite{SchmidtSTVZ21} carry over to regular tree languages, even
under some structural changes of the tree. In turn they also hold for
visibly pushdown languages under limited change operations. For
realtime deterministic context-free languages and visibly pushdown
languages under more general change operations better work bounds than
for context-free languages could be shown.

There are plenty of questions for further research, including the following: are there other
relevant change operations for trees that can be handled with work
\neps? What are good bounds for further operations? Can the bounds for
context-free languages be improved? Can the $\bigO(n^{3+\epsilon)}$ be
shown for arbitrary (not necessarily realtime) DCFLs?  
And the most challenging: are there further lower bound results that 
complement our upper bounds?

\bibliography{bibliography}

\appendix

\newpage
\section{Appendix for \autoref{subsection:labelchanges-efficient}}
\subparagraph*{Label changes: basics}\label{subsection:labelchanges-basics}

The most basic change operation for trees that we consider in this
paper is $\textsc{Set}_\sigma(u)$ which sets the label of a node $u$
to symbol $\sigma$. In this subsection, we describe how membership in a regular tree language
can be maintained under label changes, in principle, i.e., in a work
inefficient way. 

In the following, we fix a DTA
$\mathcal{B}=(Q_\calB,\Sigma, Q_f, \delta,
\calA)$, with  DFA $\calA=(Q_\calA,Q_\calB,\delta_\calA,s)$.

To maintain membership of a tree $t$ in $L$, our dynamic program will
make use of several functions, some of them related to the structure
of the tree and some of them related to the behaviour of $\calA$ and
$\calB$ on $t$. For the moment, we do not specify which (parts) of
these functions are actually stored by the program and which need to
be computed by sub-programs. In fact, this is where the inefficient 
and the work efficient program differ.

Besides the unary function $\rho_t$ already defined in
\autoref{section:preliminaries}, we use the following two functions related to the behaviour of the automata. 
\begin{itemize}
    \item  The ternary function $\calB_t: Q_\calB \times V\times V\mapsto Q_\calB$ maps
        each triple $(q,u,v)$ of a state $q \in Q_\calB$ and nodes of $t$, where $u$ is a proper ancestor of
        $v$, to the state that the run of $\calB$ on $\tree{u}[v]$  takes at
        $u$, with the provision that the state at $v$ is $q$.
    \item  The ternary function $\calA_t: Q_\calA \times V\times V\mapsto Q_\calA$ maps
        each triple $(p,u,v)$ of a state $p \in Q_\calA$ and nodes of $t$, where $u\prec v$ are siblings, to
        the state that the run of $\calA$ on $u,\ldots,v$, starting from
        state $p$,  takes after $v$.
\end{itemize}
We refer to $\calA_t$, $\calB_t$ and $\rho_t$ as
\emph{automata functions}.

We write $\calA_t(\cdot,u,v)$ and $\calB_t(\cdot,u,v)$  for the unary
functions $Q_\calA\to Q_\calA$ and $Q_\calB\to Q_\calB$, respectively,
induced by fixing two nodes $u$ and $v$.

By definition, $t\in L(\calB)$ holds exactly, if $\rho_t(r)\in Q_f$ holds for the
root $r$ of $t$.

We next describe the functions that are related to the structure of
the tree. We refer to them as \emph{basic tree
  functions}:\footnote{This term shall also cover the
  mentioned relations, i.e., Boolean functions.}

\begin{itemize}
  	\item unary functions \parent, \lsibling, \rsibling,
                  \fchild, \lchild, yielding  the parent, left and
                  right sibling, and 
                  the first and the last child of $x$,
                  respectively. They are undefined if such a node does
                  not exist;
             \item a unary function \nsiblings, yielding the number of smaller siblings of $x$.
          	\item a binary function \child: \child$(x,k)$ yields
                  the $k$-th child of $x$; 
                  	\item binary relations \anc and \ancself:
                          $\anc(x,y)$ and $\ancself(x,y)$
                          holds for $x,y\in V$ if $x$ is a strict and non-strict
                          ancestor of $y$, respectively;
          	\item a binary function \ancchild: $\ancchild(x,y)$
                          yields the ancestor of $y$ which is a child 
          of $x$;
          	\item a binary function \ancindex: $\ancindex(x,y)$
                          yields 
          the position of the ancestor of $y$ in the children sequence
          of $x$;
        \item  a binary function \lca: $\lca(x,y)$ yields the lowest common ancestor of $x$ and $y$;
        \item binary relations $\prec$ and $\preceq$: $x\prec y$ holds for $x,y\in V$ if $x$ is a left sibling of $y$. Similarly, $x\preceq y$ holds if $x\prec y$ or $x=y$.
\end{itemize}

Algorithms~\ref{procedure:evaluate_state}, \ref{procedure:evaluate_sequence} and \ref{procedure:evaluate_path}  describe straightforward procedures to
update the automata functions
 after label changes.     With these algorithms the following result
 follows immediately.
 \LemmaUpdateTuple*
              
              As a corollary we get the following result.
              \begin{proposition}
             For each regular tree language $L$,    there is a parallel constant time dynamic algorithm
                for $\RegTreeM(L)$  with work $\bigO(n^2)$ on an EREW PRAM.
              \end{proposition}
              \begin{proof}
                A dynamic algorithm can maintain $\calB_t$ and
                $\calA_t$, for all possible tuples. Since each tuple
                update only requires constant time and all these
                evaluations can be done in parallel, the claim
                follows.

              \end{proof}

	\begin{algorithm}
		\begin{algorithmic}
			
			\Procedure{\evaluatestate}{$v,\sigma,x$}
			\If{$\ancself(x,v)$} \Comment{$v$ only influences $x$ if $x$ is an ancestor of $v$}
                        \State $p \gets
                        \calA_t(s,\fchild(v),\lchild(v))$
                        \Comment{This is $s$, if $v$ has no children}
            \State $q \gets \delta(p,\sigma)$
                                             \State    \Return
                                             $\calB_t(q,v,x)$
                                             \Comment{If $v=x$ this is
                                             $q$}
			\Else
				\State \Return $\rho_t(x)$ 
			\EndIf
			\EndProcedure
			\caption{Returns the new value $\calB'(x)$, if
                          the label of node $v$ is set  to $\sigma$.}
			\label{procedure:evaluate_state}
		\end{algorithmic}
              \end{algorithm}

	\begin{algorithm}
		\begin{algorithmic}
              \Procedure{\evaluatesequence}{$v,\sigma,p,x,y$}
              \State $u \gets \parent(x)$ \Comment{If $x$ is a root
                then $u$ is the super-root...}
			\If{$\anc(u,v)$ AND $\ancindex(u,v)
                          \ge \ancindex(u,x)$ AND  $\ancindex(u,v)
                          \le \ancindex(u,y)$}
                        \State $w \gets \child(u,\ancindex(u,v))$
                             \State $p_1 \gets
                        \calA_t(p,x,\lsibling(w))$
                        \State $p_2 \gets
                        \delta_\calA(p_1,\evaluatestate(v,\sigma,w))$
                            \State \Return
                        $\calA_t(p_2,\rsibling(w),y)$               
			\Else
				\State \Return $\calA_t(p,x,y)$
                                \Comment{In this case the change at
                                  $v$ is irrelevant}
			\EndIf
			\EndProcedure
			\caption{Returns the new value for
                          $\calA_t(p,x,y)$, for siblings $x\preceq y$, if
                          the label of node $v$ is set  to $\sigma$.}
			\label{procedure:evaluate_sequence}
		\end{algorithmic}
              \end{algorithm}

	\begin{algorithm}
		\begin{algorithmic}
                  \Procedure{\evaluatepath}{$v,\sigma,q,x,y$}
                  	\If{$\anc(y,v)$ OR NOT $\anc(x,v)$}
		\State	\Return $\calB_t(q,y,x)$
			
		\Comment{If $v$ is a descendant of $y$ or not a
                  descendant of $x$ nothing changes}
                \EndIf
                
                \State $u\gets $ \lca$(v,y)$\Comment{$u$ is on the
                  path from $y$ to $x$}

                \If{$u=v$}
                               \State  $z \gets \ancchild(u,y)$
                \State $p_1 \gets \calA_t(s,\fchild(u),\lsibling(z))$
               \State $p_2 \gets \delta_\calA(p_1,\calB_t(q,y,z))$
               \State $q' \gets \delta(\calA_t(p_2,\rsibling(z),\lchild(u)),\sigma)$

                 \Else
                \State  $z \gets \ancchild (u,y)$
               \State $z' \gets \ancchild (u,v)$
                 \If{$z\prec z'$} \Comment{$v$ is to the right  of $y$ in $t$}
                \State $p_1 \gets \calA_t(s,\fchild(u),\lsibling(z))$
               \State $p_2 \gets \delta_\calA(p_1,\calB_t(q,y,z))$
               \State $p_3 \gets \calA_t(p_2,\rsibling(z),\lsibling(z'))$
               \State $p_4 \gets \delta_\calA(p_3,\evaluatestate(v,\sigma,z')$
              \State $q' \gets \delta(\calA_t(p_4,\rsibling(z'),\lchild(u)),\label(u))$
              \Else \Comment{If $z'\prec z$ the algorithm proceeds analogously}
            \EndIf
            \EndIf
            
              \State \Return $\calB_t(q',u,x)$
			\EndProcedure
			\caption{Returns the new value for
                          $\calB_t(q,y,x)$, for an ancestor $x$ of a
                          node $y$ if
                          the label of node $v$ is set  to $\sigma$.}
			\label{procedure:evaluate_path}
		\end{algorithmic}
              \end{algorithm}

\newpage
\subparagraph*{Label changes: a work-efficient algorithm}
For a non-tree zone $S$ as above, we call  $u$ its \emph{left boundary node} $\zleft(S)$ and $v$ its \emph{right boundary node} $\zright(S)$. For a tree zone $S$, we call $v$ its \emph{upper boundary node} $\zupper(S)$ and for any incomplete zone $S$, we call $w$ its \emph{lower boundary node} $\zlower(S)$. For convenience, we also call the upper-most ancestor of a lower boundary node $w$ in a non-tree zone an  \emph{upper boundary node} $\zupper(S)$, even though it is not literally on the boundary of the zone. A node is a \emph{top node} of a zone if its parent is not in the zone.

     For an incomplete zone $S=\tree[u]{v}[w]$ we denote by $\hat{S}$ the complete zone $\tree[u]{v}$. For a complete zone $S$, $\hat{S}$ is just $S$ itself. This equally applies to tree zones, i.e., if $u=v$.

     \LemmaPartitionStep*
\begin{proof}
  We call a node \emph{$S$-large} if it has more  than $\frac{1}{2}m$ descendants in $S$ (including the node itself), otherwise we call it \emph{$S$-small}. 

  We first observe that there exists siblings $u\preceq v$, for which $\tree[u]{v}\cap S$ has at least $\frac{1}{2}m$ nodes, but no node between $u$ and $v$  (including $u$ and $v$) is $S$-large. To see this, we first consider the left boundary $u_1$ and the right boundary $v_1$ of $S$ (or $u_1=v_1$ as the root, if $S$ is a tree-zone). By assumption, $\tree[u_1]{v_1}\cap S$ has at least $\frac{1}{2}m$ nodes. If there are no $S$-large nodes, we are done. Otherwise,  we choose a $S$-large node $w_1$ between $u_1$ and $v_1$ and continue with its leftmost and rightmost child as $u_2$ and $v_2$. Since $w_1$ has \emph{more} than $\frac{1}{2}m$ descendants in $S$, $\tree[u_2]{v_2}\cap S$ has at least $\frac{1}{2}m$ nodes. This process will necessarily yield $u_a\preceq v_a$ with the desired property, for some $a$. 

Let now the nodes from $x_1=u_a,\ldots,x_k=v_a$ the sibling sequence from $u_a$ to $v_a$ and let $t_1,\ldots,t_k$ be their induced trees in $S$.  For $i\le k$, let $D_i$ denote the forest  $t_1,\dots, t_i$.

  Since each $t_j$ has at most $\frac{1}{2}m$ nodes, there exists an index $j$, such that
	$\frac{1}{2}m \leq |D_j| \le m$
	holds. The subforest $S_1 \coloneqq D_j$ is therefore a zone with the required size.

        We now show how $S$ can be partitioned into at most 5 valid zones, including $S_1$.

       First, we can split  $S$ into three parts, $S_1$, $S_2 \coloneqq D_k \setminus D_j$ and $S' \coloneqq S \setminus D_k$, where $S_2$ or $S'$ can be empty. $S_2$ has at most one vertical connection node and is therefore a zone, but $S'$ need not be a zone. 

       In the latter case,  $S'$ is non-empty and has two nodes, $v_S$ and $w_{a-1}$, whose children are not in $S'$.  In this case, $S'$ can be split into at most three zones as follows.
        
        If $v_S$ belongs to a different connected component of $S'$ than $w_{a-1}$, then $S_3$ can be chosen as the tree of $S'$ that contains $D_k$ minus $D_k$, $S_4$  as the union of the trees to its left and $S_5$  as the union of the trees to its right. 

   Otherwise,   let $z$ be the lowest common ancestor of $w_{a-1}$ and $v_S$ with children $z_1,\dots, z_g$. Let $Z_j$ denote the set of all descendants of $z_1,\dots, z_j$ in $S'$. By construction, some $Z_j$ contains exactly one of the nodes $w_{a-1}$ and $v_S$. We can thus choose $S_3\coloneqq Z_j$, $S_4 \coloneqq Z_g \setminus Z_j$ and $S_5\coloneqq S' \setminus Z_g$. It is again possible that some of these zones are empty.
      \end{proof}

      \PropPartitionHierarchy*
          \begin{proof}
        We first show that a zone $S$ of some size $s$ can always be partitioned into at most $10s^\theta$ zones of size at most $s^{1-\theta}$. This can be shown by repeated application of \autoref{lemma:pre-partition-step} with $m=s^{1-\theta}$.  Starting with the partition that only consists of $S$, as long as there is a zone that is larger than $m$, one such zone can be chosen and split into at most 5 zones. This yields one more zone of size at most $m$. However, as this zone also has at least $\frac{1}{2}m$ nodes,  after $2s^\theta$ rounds, no large zone can be remaining. 
By applying this argument recursively, starting from the zone consisting of the whole tree,  a $(10,\theta)$-bounded partition hierarchy of height at most $h=\lceil \theta\rceil$ will be produced.        
\end{proof}

\LemmaComputeFromStored*
              \begin{proof}
                We first describe, how $\rho_t(v)$ can be determined, for an arbitrary node $v$. If $v$ is a leaf, this can be done directly. For non-leaves $v$, we inductively define a sequence of nodes $v_1,\ldots,v_j$ as follows.
                \begin{itemize}                 
                \item $v_1=v$.
                \item If the $i$-zone $S_i$ of $v_i$ is incomplete and its lower boundary $w$ is a descendant of $v_i$ then $v_{i+1}=w$. Otherwise $v_{i+1}$ is undefined.
                \end{itemize}
                Let $j$ be maximal such that $v_j$ is defined. Since the overall zone has level $h$ and is complete, it holds $j\le h$. By definition, for each $i<j$, the value of $\calB_t(q,v_{i+1},v_i)$ is stored, for every state $q$, because both nodes are in the same 1-zone, if $i=1$, or because both nodes are lower boundary nodes of $i$-zones of $S_{i+1}$ and $v_{i+1}$ is the lower boundary node of $S_{i+1}$. The state $\rho_t(v_j)$ can be determined from $\calA_t(s,\fchild(v_j), \lchild(v_j))$, since all children of $v_j$ are in $S_j$. The state of $v$ can now be determined, by combining the state of $v_j$ with appropriate transitions of the form $\calB_t(q,v_{i+1},v_i)$.

                We next describe, how  $\calA_t(p,u,u')$ can be determined for arbitrary pairs $u\prec u'$ of siblings of $t$. To this end, let, for each $\ell$, $S_\ell$ denote the unique zone of level $\ell$ in the partition hierarchy that contains $u$. Likewise, for  $S'_\ell$ and $u'$. Let $j$ be minimal, such that $S_j=S'_j$ holds. Since $S_h=S'_h$ must be $t$, such a $j$ exists. If $j=1$ holds, then $u$ and $u'$ are in the same zone and thus $\calA_t(p,u,u')$ is stored.

               Otherwise,  we inductively define a sequence of nodes that ``connects'' $u$ and $u'$, as follows.
               Let $x_1=u$ and $y_1$ be the rightmost sibling of $u$ in $S_1$. For each $i<j$, starting from $i=2$, let $x_i$ be the right sibling of $y_{i-1}$ and $y_i$ be the rightmost sibling of $u$ in $S_i$. We assume here, without loss of generality\footnote{The arising cases can be handled in a similar straightforward way.}, that $y_{i-1}$ is not yet the rightmost sibling of $u$ in $S_i$.
               Let $x_j$ be the right sibling of $y_{j-1}$, a left boundary node of a zone of level $j-1$. 
               From $u'$, we can define an analogous sequence from right-to left, starting with $y'_1=u'$ and $x'_1$ as the left-most sibling of $u'$ in $S'_1$. The last (left-most) node of that sequence is a node $y_j$, which is a right boundary node of a zone of level $j-1$. The behaviour of $\calA$ between $u$ and $u'$ can now be induced from $\calA_t(\cdot,x_i,y_i)$ and $\calA_t(\cdot,x'_i,y'_i)$, for $i\le j$.  Since $h$ is a constant, $\calA_t(p,u,u')$ can thus be computed in constant time, if this holds for $\calA_t(\cdot,y_j,y'_j)$ and all functions $\calA_t(\cdot,x_i,y_i)$ and $\calA_t(\cdot,x'_i,y'_i)$. 
               In fact, whenever $S_i$ is complete or the lower boundary $w_i$ is not in $\tree[x_i]{y_i}$, the information about $\calA_t(\cdot,x_i,y_i)$  is stored by the algorithm.
               Otherwise, $\calA_t(\cdot,x_i,y_i)$ can be inferred from $\calA_t(\cdot,x_i,\lsibling(v_i))$, $\calB_t(v_i)$ and $\calA_t(\cdot,\rsibling(v_i),y_i)$, where $v_i$ is the ancestor of $w_i$ that is between $x_i$ and $y_i$. $\calA_t(\cdot,x'_i,y'_i)$ can be computed analogously.

               The computation of  $\calB_t(q,u,u')$ is along similar lines. Given $u$ and $u'$, we define the zones $S_\ell$, $S'_\ell$ and the level $j$, where $S_j=S'_j$ holds, as above.
               The sequence starts again from $x_1=u$, but instead of rightmost siblings of $u$ in a zone, let $y_i$ be the highest ancestor of $u$ in the zone and instead of the right neighbour of a node $y_{i-1}$, its parent is chosen as $x_i$. Similarly, for the sequence starting from $u'$ towards $u$.
For all $i\le j$, $\calB_t(\cdot,x_i,y_i)$ is stored by the algorithm or it can be computed  $\calB_t(\cdot,u,\ancchild(w',u))$, $\calB_t (\cdot,w',v)$, $\calB_t(\ancchild(w',w))$, and the information about the behaviour of $\calA$ on the children of $w'$, excluding $\ancchild(w',u)$ and $\ancchild(w',w)$, where $w$ is the lower boundary of $S_i$ and $w'_i$ is $\lca(w,u)$. 

We emphasise that the computation of $\calA_t(\cdot,u,u')$ might use information on $\rho_t(x)$ for arbitrary nodes $x$, and the computation of  $\calB_t(\cdot,v,v')$ might use information on $\rho_t(y)$ for arbitrary nodes $y$ and $\calA_t(\cdot,z,z')$ for arbitrary $z,z'$. However, computations of $\rho_t(x)$ only use $\calA_t$ and $\calB_t$ for special tuples and computations of $\calA_t$ do not use any values of $\calB_t$. 
             \end{proof}

             \begin{proof}
               To achieve the stated bound, we use the above described algorithm with work parameter $\theta=\frac{\epsilon}{2}$. The algorithm uses a $(\theta,10)$-bounded partition hierarchy, which exists thanks to \autoref{proposition-partition-hierarchy}.
               
               As indicated before, the algorithm stores $\calA_t(\cdot,u,v)$ and $\calB_t(\cdot,u,v)$, for all special pairs $(u,v)$. As already observed before, these values only depend on the labels of the nodes of the zone relative to which $(u,v)$ is special. Therefore, if a node label is changed for some node $x$, values   $\calA_t(\cdot,u,v)$ and $\calB_t(\cdot,u,v)$ need only be updated for special pairs of zones in which $x$ occurs. Since each node occurs in exactly $h$ zones, each zone has $\bigO(n^{2\theta})=\neps$ special pairs, and the update of a single value is possible in constant time and work, thanks to \autoref{lemma-update-tuple} and \autoref{lemma-compute-from-stored}. 
             \end{proof}

\newpage
\section{Appendix for \autoref{subsubsection:structuralchanges-high-level}}
\LemmaPartitionEfficient*
\begin{proof}
  The idea is to follow the plan of \autoref{lemma:pre-partition-step}, but in constant time with little work. The first step is to identify nodes $u,v$ as in that proof.
  
  The algorithm tries to identify a zone $Z$ of some level $\ell$ such that (1) $|\hat{Z}\cap S|> \frac{1}{2}m$, (2)   $|\hat{Z'}\cap S|\le \frac{1}{2}m$ for all sub-zones of $Z$ at level $\ell-1$, and (3)  $|\hat{Z''}\cap S|\le \frac{1}{2}m$ for all child zones of $Z$ at level $\ell$. If such a zone does not exist, there must be a zone $Z$ of level 3 with $|\hat{Z}\cap S|>\frac{1}{2} m$.

  The zone $Z$ can be found by starting from $i=h$ and the overall zone $Z_h$. Then $Z_{i-1}$ is always chosen as a zone in $Z_i$ with  $|\hat{Z_i}\cap S|> \frac{1}{2}m$ that has no child zone $Z'_i$ with $|\hat{Z'_i}\cap S|> \frac{1}{2}m$. If this process stops at a zone $Z_i$ with $i>3$, $Z_i$ fulfils conditions (1) - (3). Otherwise it ends with a  zone $Z$ of level 3 with $|\hat{Z}\cap S|>\frac{1}{2}m$.

  Let us consider the case $i>3$ first. 
  If $\zlower(Z)$ has more than $\frac{1}{2}m$ descendants in $S$, then $\bigcup_{i=1}^k \hat{Z^i}$, where $Z^1,\ldots,Z^k$ are the child zones of   $\zlower(Z)$, contains at least $\frac{1}{2}m$ nodes from $S$. And thanks to (3),  $|\hat{Z^j}\cap S|\le m$, for each $j$ holds.

 Otherwise, for the top zones   $Z^1,\ldots,Z^p$ of $Z_i$ an analogous statement holds: $\bigcup_{i=1}^p \hat{Z^i}$ has at least $\frac{1}{2}m$ nodes from $S$, and $|\hat{Z^j}\cap S|\le \frac{1}{2}m$ holds, for each $j$.

 In either case $u$ can be chosen as $\zleft(Z^1)$ and $v$ as $\zright(Z^k)$.

 If $i=3$ then $u$ and $v$ can be chosen as the first and  the last child of  $\zlower(Z)$, if $\zlower(Z)$ has more than $\frac{1}{2}m$ descendants in $S$. Otherwise, it can be found by inspecting all children sets in $Z$. The identification of $Z$ requires work $\ntheta[2]$ and the computation of $u$ and $v$ at most $\ntheta[6]$, by the case $i=3$.

 The rest of the algorithm of  \autoref{lemma:pre-partition-step} can be implemented in a straightforward fashion: If $i>3$, $D_j$  can be determined by finding the  left-most top node $v'$ of $Z^1,\ldots,Z^k$ for which $|\tree[u]{v'}\cap S|\ge \frac{1}{2}m$. This can be done by recursively determining the zone in which the threshold $\frac{1}{2}m$ is reached. In the case $i=3$ it is even easier. That case dominates the work:  $\ntheta[6]$.

 The above algorithm frequently computes  $|\hat{Z}\cap S|$ for various zones $Z$. To this end, it assumes the availability of the function $\numdesc(u,v)=|\tree[u]{v}|$. We will establish in the next subsection that this function is indeed available.
 
  The size of any zone $S$ with left and right boundary nodes $\ell$ and $r$ and lower boundary $v$ can then be easily computed as $|S|=\numdesc(\ell,r)-\numdesc(v)+1$. We also have $|\hat{S}|=\numdesc(\ell,r)$.
  
  We consider three cases. If $Z$ is contained in $\hat{S}$ we have $|\hat{Z}\cap S|= |\hat{S}|-|\hat{Z}|$. If $S$ is contained in $\hat{Z}$ it is clear that $|\hat{Z}\cap S|=|S|$. In the remaining case, $S$ and $Z$ share an interval of top nodes, i.e., $\zright(S)$ is a (not necessarily immediate) right sibling of $\zleft(Z)$ (or the other way around). Then, $|\hat{Z}\cap S|$ can be determined from $\numdesc(\zleft(Z), \zright(S))$, $\numdesc(\zleft(S), \zright(Z))$, $\numdesc(\zlower(S))$ and $\numdesc(\zlower(Z))$, depending on the relative positions of these nodes.

  Altogether, the algorithm works in parallel constant time with  work $\bigO(n^{6\theta})$.
\end{proof}

\newpage
\section{Appendix for \autoref{subsubsec:treefunctions}}
In the following, we first describe which functions are used and for which tuples which function values are stored. Then we show that
\begin{itemize}
\item arbitrary function values can be efficiently computed from the stored function values,
\item the stored function values can be efficiently updated, and
\item initial stored function values can be efficiently computed, for new zones.  
\end{itemize}

In the following, we denote the  zone on level $\ell$ that contains a node $v$ by  $Z^\ell(v)$ and likewise the zone on level $\ell$ that contains a zone $Z$ by $Z^\ell(Z)$.

Besides the tree functions defined in \autoref{subsection:labelchanges-basics}, we use the following tree functions:
\begin{itemize}
\item a unary function \nchildren that yields the number of children of a node $x$;
\item a binary  function \numdesc\ such that $\numdesc(u,v)$ is the number of nodes in \tree[u]{v}, for siblings $u\preceq v$. We write $\numdesc_3(u,v)$ for the number of such nodes  that are inside $Z^3(v)$. We write \numdesc$(v)$ for \numdesc$(v,v)$. This function was needed in the proof of \autoref{lemma-partition-efficient}.
\end{itemize}
We lift functions and relations to zones. In the following, $Z_1$ and $Z_2$  always have the same level $\ell$ in the partition hierarchy. 
\begin{itemize}
\item $\anc(Z_1,Z_2)$ holds, if
  $\anc(u,u_2)$ holds for the lower boundary $u$ of $Z_1$ and all
  $u_2\in Z_2$;
\item If $\anc(Z_1,Z_2)$ holds, then
  $\ancchild(Z_1,Z_2)$ is the child zone $Z$ of $Z_1$ with
  $\ancself(Z,Z_2)$;
\item   $\lca(Z_1,Z_2)$ is the lowest zone
  $Z$ of level $\ell$ with $\ancself(Z,Z_1)$ and $\ancself(Z,Z_2)$;
\item $Z_1 \prec Z_2$ holds if $u_1\prec u_2$ holds  for all top nodes $u_1\in Z_1$ and $u_2\in Z_2$.
\end{itemize}

  Finally, there are some additional counting functions for zones, that do not mirror any  tree functions and mostly refer to the surrounding zone of the next level:
\begin{itemize}
\item For a zone $Z$ of level $\ell$, $\ntops(Z)$ yields the number of
  top nodes in $Z$, $\nzleft(Z)$ the number of sibling nodes of
  $Z$'s top nodes, that are to the left of $Z$ in
  $Z^{\ell+1}(Z)$, and $\nzdesc(Z)$ the number of nodes below $Z$ in
  $Z^{\ell+1}(Z)$. 
\item For two zones $Z_1\preceq Z_2$ of the same level
  $\ell$ with $Z^{\ell+1}(Z_1)=Z^{\ell+1}(Z_2)$, $\zsize(Z_1,Z_2)$ is
  the number of nodes in $Z^{\ell+1}(Z_1)$ that are inside or below any zone
  $Z$ with $Z_1\preceq Z \preceq Z_2$.
\end{itemize}

The algorithm stores the following information, where all zones belong  to the current primary $3$-pruned partition hierarchy $H$ of $t$.
All zones always means all operable zones that occur in the extended partition hierarchy.

\begin{itemize}
\item For all nodes $v$ of the tree, the algorithm stores $\parent(v)$, $\fchild(v)$, $\lchild(v)$, $\lsibling(v)$, $\rsibling(v)$, $\nchildren(v)$ and $\nsiblings(v)$.
\item For all pairs $(u,v)$ of nodes inside the same lowest zone  $Z^3(u)=Z^3(v)$, $\anc(u,v)$, $\ancchild(u,v)$, $\lca(u,v)$, $\numdesc_3(u,v)$, and whether $u\prec v$ holds, is stored. All these tuples are associated with zone $Z^3(u)$.
\item For all  zones $Z$, $\ntops(Z)$, $\nzleft(Z)$ and $\nzdesc(Z)$ are stored. These tuples are associated with $Z$.
\item For all  $\ell\ge 3$ and all pairs $(Z_1,Z_2)$ of zones of level $\ell$ with $Z^{\ell+1}(Z_1)=Z^{\ell+1}(Z_2)$ the function values $\lca(Z_1,Z_2)$, $\zsize(Z_1,Z_2)$ and $\ancchild(Z_1,Z_2)$ and whether $\anc(Z_1,Z_2)$ and $Z_1\prec Z_2$ hold, is stored. These tuples are associated with $Z^{\ell+1}(Z_1)$.
\item The algorithm stores $\calA_t(p,u,v)$ for each state $p$ of $\calA$ and each special horizontal pair $u,v$.  Furthermore, it stores $\calB_t(q,u,v)$, for each state $q$ of $\calB$ and each special vertical pair $u,v$. Here, special pairs are defined as in \autoref{subsection:labelchanges-efficient}, with the provision that all pairs $u,v$ of nodes of the same zone $S$ of level 3 with $u\prec v$ or $\anc(u,v)$ are special (horizontally or vertically), associated with $S$, and that all pairs $(\zleft(Z),\zright(Z))$ are special horizontally and all pairs $(\zupper(Z),\zlower(Z))$ special vertically, for all zones $Z$, and they are associated with $Z$.
\end{itemize}

\LemmaTreeFunctions*
\begin{proof}
	We first describe a $\ntheta$ work algorithm to compute $\child(v,k)$. 
	At first, it can be checked if $v$ has at least $k$ children using $\nchildren$, since otherwise, the value is undefined. If $v$ is an inner node in $Z^3(v)$, all its children are in $Z^3(v)$. Therefore, $\child(v,k)$ can be determined by testing whether $\childindex(v,u)=k$ holds, for each of the $\bigO(n^{3\theta})$ nodes $u$ in $Z^3(v)$ in parallel.
	Otherwise, if $v$ is a leaf  in $Z^3(v)$, let $\ell < h$ be maximal such that $v$ is a lower boundary node of $Z^\ell(v)$. The child zone $Z$ of $Z^\ell(v)$ with $\nzleft(Z)<k\le \nzleft(Z)+\ntops(Z)$ can be identified by inspecting all zones in $Z^{\ell+1}(v)$. Clearly, this zone contains $\child(v,k)$. Starting from $Z'_\ell=Z$, in $\ell-3$ steps, zones $Z'_i$ of lower levels containing $\child(v,k)$ as top node can be identified, for  $\ell>i\ge 3$, until  the lowest level zone $Z'_3$ containing  $\child(v,k)$  is reached. Finally, $\child(v,k)$ can be determined in $Z'_3$. In either case, the work is bounded by $\bigO(n^{3\theta})$.
       
    To evaluate the other tree functions, we use the following two simple observations. First, for each node $u$ and level $\ell$ it can be determined whether $u$ is an ancestor of $u_\ell=\zlower(Z^\ell(u))$ in constant sequential time by testing if $\ancself(u,\zlower(Z^3(u))$ holds and $\ancself(Z^i(u), Z^i(\zlower(Z^{i+1}(u)))$ hold for all $3 <  i < \ell$.  All this information is stored. Secondly, it can be determined in sequential constant time whether a node $u$ is a top node in $Z^\ell(u)$ by checking if $Z^\ell(u)\neq Z^\ell(\parent(u))$ holds.
    
    We now show that all other tree functions can be evaluated with constant work, given the stored information.
 	The relations $\ancself$ and $\preceq$ can easily be computed from the function $\anc$ and $\prec$ with an additional equality test.

        Let $u,v$ be any nodes and let $j$ be the lowest level such that $Z^j(u)= Z^j(v)$ holds.
        If $j=3$, $u\preceq v$, $\ancself(u,v)$, $\lca(u,v)$ and $\ancchild(u,v)$ are stored.
		Otherwise, $\anc(u,v)$  holds, if and only if  $\anc(Z^{j-1}(u),Z^{j-1}(v))$ holds and $u$ is an ancestor 
	 	of $\zlower(Z^{j-1}(u))$. Likewise, $u\prec v$ holds if and only if $Z^{j-1}(u) \prec Z^{j-1}(v)$ holds and $u$ and $v$ are both top nodes in $Z^{j-1}(u)$ and $Z^{j-1}(v)$, respectively. 
	 	Further, $\lca(u,v)=\zlower(\lca(Z^{j-1}(u), Z^{j-1}(v)))$. 
		Obviously, all these computations can be done with constant work.

     To  compute $\ancchild(u,v)$ if $j>3$, it can be tested in constant sequential time if $\anc(u,v)$ holds, as above. If so and $u\neq\zlower(Z^3(u))$ holds, then $\ancchild(u,v)=\ancchild(u,\zlower(Z^3(u))$. If $\anc(u,v)$ and $u=\zlower(Z^3(u))$ hold, let $i$ be the highest level with $u=\zlower(Z^i(u))$. If $i +1 < j$ holds, $v$ is a descendant of $\zlower(Z^i(u))$ and therefore\footnote{We recall that $\zupper(Z)$ is defined for tree zones \emph{and} non-tree zones.} $\ancchild(u,v)=\zupper(Z)$ with $Z=\ancchild(Z^i(u), Z^i(u_{i+1}))$. If $j=i+1$ holds, $\ancchild(u,v) =\zupper(Z)$ for $Z=\ancchild(Z^i(u), Z^i(v))$.
	
	The index $\ancindex(u,v)$ can be simply obtained as $\nsiblings(\ancchild(u,v))+1$.

	For $\numdesc(u,v)$, let $\numdesc^{t}(Z)=\sum_{i=\ell}^{h-1} \numdesc(Z^{i+1}(Z))$ denote the total number of descendants in $t$ of a zone $Z$ of level $\ell$ . For two zones $Z_1 \preceq Z_2$ of level $\ell$ that are in the same zone $Z$ on level $\ell+1$, we denote $\textsc{size}^t(Z_1,Z_2)= \zsize(Z_1,Z_2)+\numdesc^t(Z)$ if $Z_1 \preceq \ancchild(\parent(S_1), \textsc{lower}(Z))\preceq Z_2$ and else $\textsc{size}^t(Z_1,Z_2)= \zsize(Z_1,Z_2)$. Hence, $\textsc{size}^t$ gives the number of all nodes in, below or between $Z_1$ and $Z_2$ in $t$. Similarly, for two nodes $u,v$ with $Z^3(u)=Z^3(v)$, we define $\textsc{size}^t(u,v)= \numdesc_3(u,v)+\numdesc^t(Z^3(u))$ if $u \preceq \ancchild(\parent(u), \textsc{lower}(Z^3(u)))\preceq v$ and else $\textsc{size}^t(u,v)= \numdesc_3(u,v)$.
	
	For each level $\ell$, let $u_\ell = \zright(Z^\ell(u))$ and $v_\ell=\zleft(Z^\ell(v))$. Let $j$ be the lowest level where $Z^j(u)=Z^j(v)$ holds. At level $j-1$ the parent of $u$ and $v$ must be a lower boundary node and all top nodes of $Z^i(u)$ and $Z^i(v)$ for $i<j$ are siblings of $u$ and $v$.
	
	Then, we have $\numdesc(u,v)=\textsc{size}^t(u,u_3)+ \sum_{i=3}^{j-1} \textsc{size}^t(Z^i(\rsibling(u_i),Z^i(u_{i+1})))+\textsc{size}^t(Z^{j-1}(u_{j-1}),Z^{j-1}(v_{j-1}))+\sum_{i=3}^{j-1} \textsc{size}^t(Z^i((v_{i+1}),Z^i(\lsibling(v_{i}))))+\textsc{size}^t(v_3,v)$.

      \end{proof}

      \LemmaUpdateTreeFunctions*

      \begin{proof}
	It is clear that a $\Relabel$ operation does not change the shape of the input tree and thus all stored functions are unaffected. 
	
	At an $\AddChild$ operation, a leaf $v$ without right siblings is inserted.
	For the functions \parent, \lsibling, \rsibling, \fchild, \lchild, \nchildren and \nsiblings, such an operation only affects at most two function values (the one of $v$ and possibly one for its parent or left sibling). It is easy to see that these functions can be updated with constant work.
	
	The number $\nsiblings$ does not change for any node at the insertion of a leaf. So, the only affected node is the added leaf itself and it can be deduced from its left sibling, if it exists.

	The values $\anc(u,u')$,$\ancchild(u,u')$ and $\lca(u,u')$ for pairs $(u,u')$ of nodes i $u,u'$ nside the same lowest zone are only affected by the insertion of $v$, if $u'=v$ holds. It is easy to see that these new values can be deduded directly from information about $\parent(v)$ and $\lsibling(v)$ for all affected pairs in parallel and hence, computation is possible with work $\ntheta[3]$.
	
	The function $\numdesc_3(u,u')$ must be updated for tuples associated with $Z^3(v)$. It is increased by $1$ if and only if $u\preceq \ancchild(\parent(u),v)\preceq u'$ holds. There are $\ntheta[6]$ such tuples and hence, the update takes $\ntheta[6]$ work.
	
	The function $\ntops$ needs to be updated only for zones $Z=Z^\ell(v)$ for $\ell\geq 3$. There is one such zone per level.
	The value $\ntops(Z)$ needs to be increased by $1$ exactly if $Z^\ell(\parent(v))\neq Z$ holds. 
	The value $\nzleft(Z)$ is never affected by $v$, as $v$ has no right siblings.
	
	The value $\nzdesc(Z)$ needs to be incremented for all zones $Z$ of level $\ell$ in $Z^{\ell+1}(v)$ by $1$ if and only if $\anc(Z,Z^\ell(v))$ holds which is stored explicitly, resulting in $\ntheta$ work.
	
	For each level $\ell>3$ and $\ell$-zones $Z_1,Z_2$, the values $\lca(Z_1,Z_2)$, $\ancchild(Z_1,Z_2)$ and whether $\anc(Z_1,Z_2)$ and $Z_1\prec Z_2$ hold are unaffected by the update of $v$, as the structure of the partition hierarchy does not change.
	
	The value $\zsize(Z_1,Z_2)$ needs to be incremented by $1$ exactly if the tuple is associated with $Z^{\ell+1}(v)$ and $Z_1\preceq \ancchild(\parent(Z_1),Z^\ell(v))\preceq Z_2$ holds, which can be checked in constant sequential time as described above. There are at most $\ntheta[2]$ such tuples.
	 
\end{proof}

To complete this section, we show that initial stored function values can be efficiently computed for new zones.

\begin{lemma}\label{lemma-compute-zone-functions}
	Let $\theta>0$ and $H$ be a $3$-pruned partition hierarchy of $t$ with automata and tree functions. Let $S$ be a new $\ell$-zone, for some level $\ell$, partitioned into $\ntheta$ sub-zones. Then for each automata and tree function, the function values for all tuples associated with $S$, can   be computed in constant time with $\ntheta[7]$ work.
\end{lemma}
\begin{proof}
	The values of all tree and automata functions for an arbitrary tuple $(u,v)$ can be computed in sequential constant time using $H$ as shown in \autoref{lemma-compute-from-stored} for automata and in \autoref{lemma-compute-from-stored-tree-functions} for tree functions.
	
	If $\ell=3$, the functions $\anc$, $\ancchild$, $\lca$, $\numdesc_3$ and $\prec$, as well as the automata functions need to be computed for every pair of nodes in the same zone. We have $\bigO(\ntheta[7])$ such tuples and overall work.
	
	For a zone $Z$ of $S$, let $l_Z$ and $r_Z$ be the left and right boundary nodes. Then,  $\ntops(Z)=\nsiblings(r_Z)-\nsiblings(l_Z)$. If $l_{S}\preceq l_Z $, then $\nzleft(Z)=\nsiblings(l_Z)-\nsiblings(l_S)$ holds and else, $\nzleft(Z)=\nsiblings(l_Z)$.
	Similarly, let $b_Z$ be the lower boundary of $Z$. Then, $\nzdesc(Z)=\numdesc(b_Z)-\numdesc(b_S)$ if $\anc(b_S,b_Z)$ holds and else, $\nzdesc(Z)=\numdesc(b_Z)$.
	
	For two zones $Z_1,Z_2$ with left boundaries $l_1,l_2$ right boundaries $r_1,r_2$ and lower boundaries $b_1,b_2$ it holds that $\ancself(Z_1,Z_2)=\ancself(v_1,l_1)$, $Z_1\preceq Z_2$ if and only if $l_1 \preceq l_2$ and  $\lca(Z_1,Z_2)=Z^\ell(\lca(l_1,l_2))$. Additionally, $\zsize(Z_1,Z_2)=\numdesc(l_1,r_2)+\numdesc(b_S)$ if $Z_1\preceq \ancchild(\parent(Z_1),Z(b_S))\preceq Z_2$ holds, and else $\zsize(Z_1,Z_2)=\numdesc(l_1,r_2)$. This results in $\ntheta[2]$ work.
	
	For the automata functions, we also have $\ntheta[2]$ special tuples of level $\ell$ in $S$ that can be inferred with the help of $H$ as described before. Hence, computation takes 
	$\ntheta[2]$ work as well.
	
	Since $S$ is partitioned into $\ntheta$ zones, at most $\ntheta[2]$ zones require an update. An algorithm given $S$ and an array containing $S_1,\dots, S_k$ (or the nodes in $S$ if $\ell=3$) can just assign one processor to every pair of positions (or single position in the case of unary functions) in the array and compute the function value in parallel using $H$ in constant sequential time.
\end{proof}
 
\newpage
\section{Appendix for \autoref{subsec:dcfl}}
Since we only consider deterministic context-free languages in this paper, we only define deterministic pushdown automata.

 \begin{defi}
    A \emph{deterministic pushdown automaton} (DPDA) $(Q, \Sigma, \Gamma, q_0, \tau_0, \delta, F)$ consists of
    a set of states $Q$,
    an input alphabet $\Sigma$,
    a stack alphabet $\Gamma$,
    an initial state $q_0 \in Q$,
    an initial stack symbol $\tau_0 \in \Gamma$,
    a partial transition function $\delta: (Q \times (\Sigma \cup \{\lambda\}) \times \Gamma) \to (Q \times \Gamma^\ast)$ and
    a set $F \subseteq Q$ of final states.

    Here, $\delta(p,\lambda,\tau)$ can only be defined, if $\delta(p,\sigma,\tau)$ is undefined, for all $\sigma\in\Sigma$.  
    A DPDA is called realtime (RDPDA) if it does not have any $\lambda$-transitions.
\end{defi}

A configuration $(p,w,u)$ of a DPDA consists of
a state $p \in Q$,
the suffix $w \in \Sigma^\ast$ of the input word that still needs to be read and
the current stack content $u \in \Gamma^\ast$.
A run of a PDA $\calA$ on a word $w \in \Sigma^\ast$ is the uniquely determined maximal sequence of configurations $(p_1,w_1,u_1) \ldots (p_m,w_m,u_m)$
where for every $i < m$ either
\begin{itemize}
    \item $\delta(p_i,\sigma,\tau) = (p_{i+1},\gamma)$, $\sigma \in \Sigma$, $\tau \in \Gamma$, $w_i = \sigma w_{i+1}$, $u_i = \tau u_i'$ and $u_{i+1} = \gamma u_i'$ or
    \item $\delta(p_i,\lambda,\tau) =(p_{i+1},\gamma)$, $\tau \in \Gamma$, $w_i = w_{i+1}$, $u_i = \tau u_i'$ and $u_{i+1} = \gamma u_i'$
\end{itemize}
holds.
A run is \emph{accepting} if $w_m=\lambda$ and $p_m \in F$ hold.

The language $L(\calA)$ decided by a PDA $\calA$ consists of exactly the words $w$ on which there is an accepting run of $\calA$ that starts with the configuration $(s,w,\tau_0)$.

    A language $L$ is called realtime deterministic context-free if there is a (realtime) deterministic PDA $\calA$ with $L(\calA)=L$.

    For each  (realtime) deterministic PDA $\calA$ there is a  (realtime) deterministic PDA $\calA'$ with $L(\calA')=L(\calA)$ and $|\gamma|\le 2$, whenever    $\delta'(p,\sigma,\tau)=(q,\gamma)$. We will therefore assume that all DPDAs obey this restriction.  Transitions with , $|\gamma|=0$ ($1$, $2$) are called  \emph{pop transitions} (\emph{stack neutral},  \emph{push transitions}).

    \TheoDCFL*

Before we describe the proof of \autoref{theo:dcfl}, we give a glimpse
of the proof of \autoref{theo:cfl}. The dynamic algorithm in that  proof was formulated
relative to a context-free grammar $G$, but on a high-level it can
also be described with respect to a PDA $\calA$. In this view, the
algorithm maintains information about the possible behaviour of
$\calA$ on two substrings $u,v$ of the input string $w$, where $v$ is
strictly to the right of $u$. Basically that information is, for given
$p,q,\tau$, whether
there a is a 
run of $\calA$ on $u$ starting from state $p$ with a stack consisting
only of
$\tau$ that produces a stack content $s$, such that there is another
run of $\calA$ on $v$  starting from state $q$ with stack content
$s$ that empties its stack at the end of $v$. 

The program thus uses 4-ary relations (induced by the first and last
positions of $u$ and $v$ in $w$) and needs work $\bigO(n^3)$, for each
4-tuple. We do not know whether it can be significantly improved.

Our algorithm will maintain  different relations and functions, but 
some of them will
also be 4-ary. However, its work per tuple
update is only \ntheta, for some $\theta>0$ depending on
$\epsilon$. This is mainly thanks to the fact that the run of a  RDPDAs is
uniquely determined, for any given configuration. Furthermore, the
algorithm will use only \ntheta \emph{special values} for the 4-th parameter of
tuples, therefore yielding an overall work of $n^{3+\epsilon}$.

\begin{proofsketch}
  Let $\calA$ be a fixed RDPDA. By $w$ we always denote  the current
  input word. We often write substrings of $w$ as $w[i,j]$ if their
  boundary positions $i$ and $j$ matter, but often, we will refer to a
  substring $u$ of $w$ without explicit mention of these
  positions. However, such a $u$ should always be understood as
  a particular occurrence of substring at a particular position. The
  algorithm maintains a length bound $n$, as discussed in \autoref{section:preliminaries}. We
  first consider the situation where the length of $w$ remains between
  $\frac{1}{4}n$ and $\frac{1}{2}n$.
  
The dynamic algorithm for $L(\calA)$ is similar to the algorithm of \autoref{theo:cfl}, in
that it maintains information on runs of $\calA$ on strings $v$, that
start from the stack that was produced by another run on a string $u$,
which is located somewhere to the left of $v$ in the current input
word $w$.

A configuration $C=(p,u,s)$ consists of a state $p$, a string $u$ that
is supposed to be read by $\calA$ and a string $s$, the initial stack
content. 

We use the function $\topk(s)$ which yields the string of the $k$
topmost symbols of $s$.
We write $\run(C)$ for the unique run of $\calA$ starting from $C$.

We use the following functions $\dstate$, $\dstack$, and $\dempty$ to describe the
behaviour of $\calA$ on configurations.
\begin{itemize}
    \item $\dstate(C)$ yields the last state of $\run(C)$.
    \item $\dstack(C)$ yields the stack content at the end of  $\run(C)$.
         \item $\dempty(C)$ is the position in $u$, after which $\run(C)$ empties its stack. It is zero, if this does not
           happen at all.
         \end{itemize}

The algorithm will maintain the following information, for each
configuration $C=(p,u,\tau)$, where $u=w[i,j]$, for some $i\le
j$, for each suffix $v=w[m,n]$ of $w$, where $j<m$, each state $q$, and some $k\le n$.
\begin{itemize}
\item $\dall(C)$ defined as the tuple $(\dstate(C), |\dstack(C)|,
  \topk[1](\dstack(C)) ,  \dempty(C),)$, consisting of  the
  state,   the height of the stack, the top symbol of the stack,  
 at the end of the run on
  $C$ and the position where the run ends. If the run empties the
  stack prematurely or at the end of $u$, then $\topk[1](\dstack(C))$
  is undefined;
\item $\pushpos(C,k)$, defined as the length of the longest prefix
  $x$ of $u$, such that
  $|\dstack(p,x,\tau)|=k$. Informally this is the
  position of $u$ at which the $k$-th symbol of $\dstack(C)$, counted
  from the bottom,  is
  written;
  \item $\poppos(C,q,v,k)$, defined as the pair $(o,r)$, where $o$ is
    the length of the shortest prefix
  $v'$ of $v$, for which $\run(q,v, \topk(\dstack(C)))$ empties its stack,  and $r$ is the state it enters. Formally,
  $\poppos(C,q,v,k)=(\dempty(q,v,\topk(\dstack(C))),\dstate(q,v',
  \topk(\dstack(C)))$, where $v'$ is the prefix of $v$ of length
  $\dempty(q,v,\topk(\dstack(C)))$. We refer to $o$ by
  $\popposi(C,q,v,k)$ and to $r$ by  $\popposii(C,q,v,k)$.
\end{itemize}
We note that $\pushpos(C,\cdot)$ and $\poppos(C, \cdot, \cdot, \cdot)$
are undefined,  if the run of $\calA$ on
$C$ ends prematurely.

We next sketch how these functions can be updated with work $\bigO(n^5)$ after a change
operation, under the assumption that the function values are stored,
for every $k\le n$. Afterwards, we explain how work $\bigO(n^{3+\epsilon})$,
for every $\epsilon>0$, can be achieved by storing  $\poppos$ only for
special $k$.

We first explain, how $\dall(C)$ can be updated after changing  some position $\ell$
from 
a symbol $\sigma$ to a new symbol $\sigma'$. 
Clearly, the tuple needs
only be modified, if $\ell\in[i,j]$. Let
$u_1=w[i,\ell-1]$ and $u_2=w[\ell+1,j]$.
Thus, $u=u_1\sigma u_2$ and the update needs to determine 
$\dall(C')$ for $C'=(p,u',\tau)$, where $u'=u_1\sigma' u_2$.
We refer to
Figure~\ref{fig:dcfl-overview} for an illustration.

The algorithm first figures out the behaviour of
$\calA$ on $u_1\sigma'$. The state and the top stack symbol after
reading $u_1$ can be obtained as $p_1=\dstate(p,u_1,\tau)$ and
$\tau_1=\topk[1](\dstack(p,u_1,\tau))$. Let further $h=|\dstack(p,u_1,\tau)|$.

If the transition given by $\delta(p_1,\sigma',\tau_1)$ is a push or
stack neutral transition, the
state $p_2$  and the top
stack symbol $\tau_2$  after reading $\sigma'$ can be determined from
$\delta$.

In the very simple case that the run on $u_2$, starting from $p_2$ with stack symbol $\tau_2$
does not empty its stack, then its behaviour determines
$\dall(C')$. In particular,   $\dstate(C')=\dstate(p_2,u_2,\tau_2)$ if the
latter is defined.

\begin{figure}[t]
    \includegraphics[scale=0.165]{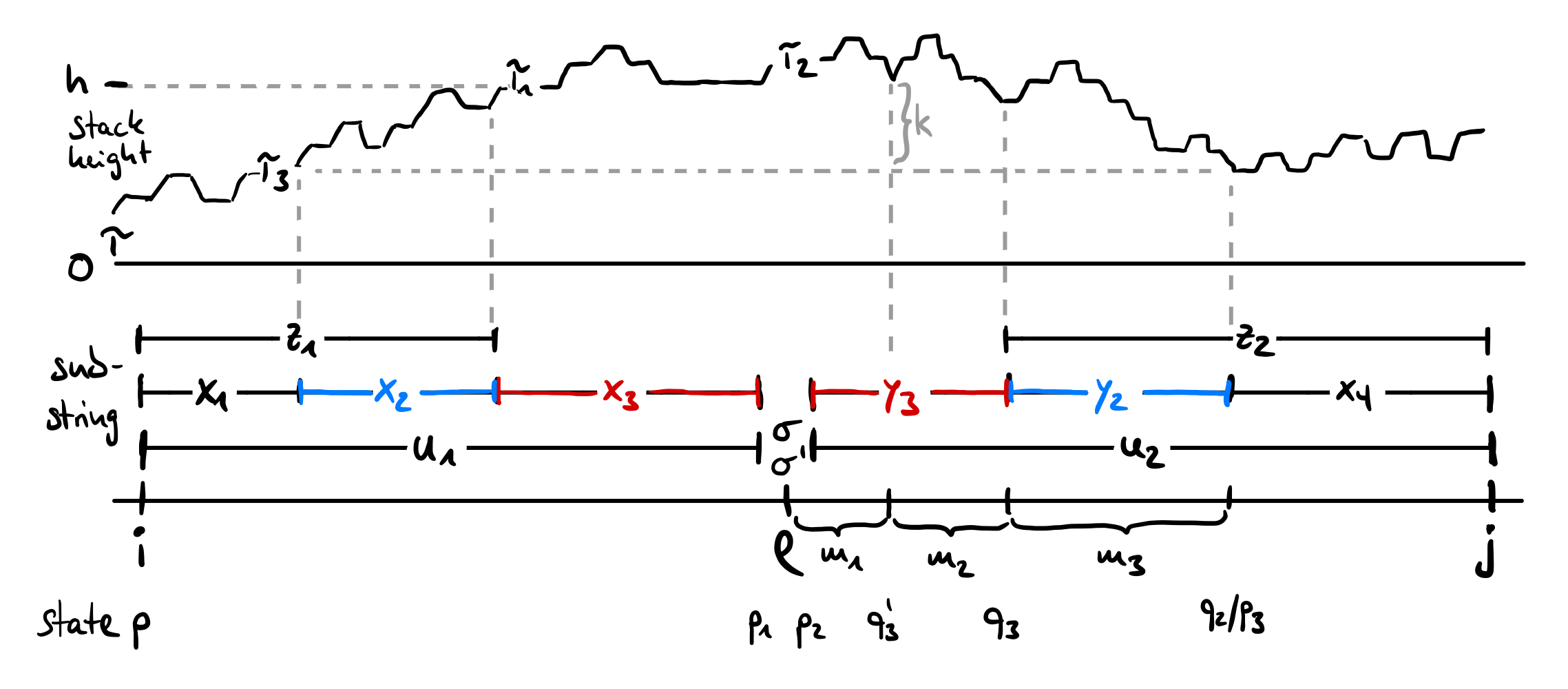}
    \caption{Illustration of the update of $\dall$ after a push
      transition at position $\ell$, i.e., Case (3).}
    \label{fig:dcfl-overview}
\end{figure}

In the general case, the algorithm decomposes $u_1$ as
$x_1x_2x_3$ and $u_2$ as $y_3y_2x_4$, where
\begin{itemize}
\item $y_3$ is the minimal prefix of $u_2$ such that
  $|\dstack(p,u_1\sigma'y_3,\tau)|=h-1$; 
\item given this $y_3$, $x_3$ is the minimal suffix of $u_1$ such that
  for the remainder $z_1$ of $u_1$ (without suffix $x_3$)
  $|\dstack(p,x_1x_2,\tau)|=h-1$ holds; 
\item $k$ is maximal such that 
  $h-k=|\dstack(p,u_1\sigma'y_3y_2,\tau)|$, for some  prefix
  $y_2$ of the remainder $z_2$ of $u_2$
  (without $y_3$);
\item $y_2$ is the minimal such prefix of $z_2$;
\item 
  $x_1$ is the maximal prefix of  $z_1$ such that
  $|\dstack(p,x_1,\tau)|=h-k$;
\item $x_2$ and $x_4$ are the remaining strings in $z_1$ and $z_2$, respectively.
\end{itemize}
Intuitively, the strings $x_3$ and $y_3$ are chosen in a minimal
fashion such that the computation of $\calA$ on the remaining parts of
$u_1$ and $u_2$ can be easily reconstructed from the stored values.    
The choice of $y_2$ guarantees that the run on 
$x_4$ is non-decreasing.  Therefore, it suffices to determine
the state $p_3$ and the top stack symbol $\tau_3$ after the run on
$x_1x_2x_3\sigma'y_3y_2$ to determine  $\dstate(C')$ as
$\dstate(p_3,x_4,\tau_3)$. The other components of $\dall(C')$ can
be determined in a similar fashion. For $|\dstack(C')|$, $h-k$ has to
be taken into account.
In the following, $q_3$ will denote the state of $\calA$ after reading
$x_1x_2x_3\sigma'y_3$.

The actual value of $y_3$ depends on the   transition type at
position $\ell$.

We distinguish three cases, depending on the transition type at
position $\ell$.

\begin{description}
    \item[Case (1)] $\delta(p_1,\sigma',\tau_1)$ is a pop transition.
      In this case, $\delta(p_1,\sigma',\tau_1)$ yields $q_3$ and  $y_3=\lambda$. 
\item[Case (2)] $\delta(p_1,\sigma',\tau_1)$ is stack neutral. Then
  the length of $y_3$ and $q_3$ can be determined as
  $\dempty(p_2,v,\tau_2)$ and $\dstate(p_2,v,\tau_2)$,
  respectively.
\item[Case (3)] $\delta(p_1,\sigma',\tau_1)$ is a push
  transition. This is the precise case, illustrated in Figure~\ref{fig:dcfl-overview}. In this case $y_3$ is the shortest prefix of $u_2$ on
  which $\calA$ deletes  $\tau_2$ and then $\tau_1$ from the
  stack. Thus the algorithm first determines
  $m_1=\dempty(p_2,v,\tau_2)$ and $q'_3=\dstate(p_2,v,\tau_2)$, where $v=w[\ell+1,n]$. Then
  it sets $m_2=\dempty(q'_3,v',\tau_1)$ and $q_3=\dstate(q'_3,v',\tau_1)$, where
  $v'=w[\ell+m_1+1,n]$. Then, $y_3$ is chosen as the prefix
  of $u_2$ of length $m_1+m_2$. 
  However, we need to take care of the special case  that the run of $\calA$
  on $(q'_3,v',\tau_1)$ does not empty its stack, in which
  case $\ell+m_1+m_2>j$. In that case $x_4$ can be chosen as
  $w[\ell+m_1+1,j]$ and $\dstate(C')$ can be determined as
  $\dstate(q'_3,w[\ell+m_1+1,j],\tau_1)$. Similarly for the other
  components of $\dall(C')$. 
\end{description}

The end of $x_2$ (and thus $x_3$) is now identified as the position, where the 
symbol at position $h-1$ of the stack $\dstack(p,u_1,\tau)$ was pushed,
which is just $\pushpos((p,u_1,\tau),h-1)$.

Let $z_1$ and $z_2$ be defined as above. The length $m_3$ of $y_2$ can
be determined as
$(m_3,q_2)=\poppos((p,z_1,\tau),q_3,z_2,k)$, for the maximal $k$, for
which $m_3\le |z_2|$. And $q_2$ yields the desired first state $p_3$ of the run
on $x_4$. The end of $x_1$ can be determined as $\pushpos((p,u_1,\tau),h-k-1)$
and the initial stack symbol $\tau_3$  for the run on $x_4$ as
$\topk[1](\dstack(p,x_1,\tau))$.

The update of $\pushpos(C',k)$ is along similar lines. In  the very simple
case that $\dempty(p_2,u_2,\tau_2)=0$ holds,  $\pushpos(C',k)$ can
be determined by inspecting $\pushpos((p_2,u_2,\tau_2),\cdot)$ and
$\pushpos((p,u_1,\tau),\cdot)$. Otherwise, it can be determined by
inspecting $\pushpos((q_3,x_4,\tau_3),\cdot)$ and
$\pushpos((p,x_1,\tau),\cdot)$.

If the position $\ell$ of the  change is within $u$, the update of
$\poppos(C,q,v,k)$ is similar, as well: in the simple case, first
$\poppos((p_2,u_2,\tau_2),q,v,\cdot)$ is considered and then,
possibly, $\poppos((p,u_1,\tau),q,v,\cdot)$; in the other case,
$\poppos((q_3,x_4,\tau_3),q,v,\cdot)$ and
$\poppos((p,x_1,\tau),q,v,\cdot)$ are relevant.

Let us finally consider the update of $\poppos(C,q,v,k)$, if a position
$\ell$ which is located inside $v$, is changed from $\sigma$ to
$\sigma'$. To this end, let $v=v_1\sigma v_2$ and $v'=v_1\sigma'v_2$
and let $m_k=\popposi(C,q,v,k)$. If
$m_k\le |v_1|$, then $\poppos(C,q,v',k)=\poppos(C,q,v,k)$. Otherwise,
the algorithm determines the largest $k'\le k$, for which $m_{k'}=\popposi(C,q,v,k')\le
|v_1|$. Let  $x_0$ be the
prefix of $v$ of length $m_{k'}$.  With the help of
$\pushpos(C,\cdot)$ and $\dall$ it figures out
$\tau'=\topk[1](\dstack(q,x_0,\dstack(C)))$. We note that this is a
symbol that is pushed by the computation on $C$ and becomes accessible
by the final pop transition of the computation on $x_0$. 
Similarly, as for 
updating $\dall$, the algorithm determines strings $x_1$ and $y_1$, such that
$x_0x_1\sigma'y_1$ is the shortest prefix of $v'$, for which the run
on $(q_{k'},x_1\sigma'y_1,\tau')$ empties its stack (and reaches some
state $q'$). The value of
$\poppos(C,q,v',k)$ can then basically be determined as
$\poppos((p,z_1,\tau),q',v'',k-k')$, where $v''$ is such that
$v'=x_0x_1\sigma'y_1v''$ and $z_1$ is the maximal prefix of $u$ that
yields a stack that is $k'$ symbols lower than at the end of
$u$. 

For each possible value of $i<j<m$ and $k$, the work of the algorithm
to update $\poppos$
is dominated by computing $(m_{k'},q_{k'})$ for all $k'\le k$. The
overall work per update is thus bounded by $\bigO(n^5)$. 

To improve the overall work, the algorithm stores $\pushpos$ and
$\poppos$ only for tuples, in which the parameter $k$ is of the form
$an^{b\theta}$, for integers $b<\frac{1}{\theta}$ and $a\le n^\theta$,
for some fixed $\theta>0$. Let us assume in the following, for simplicity, that
$g=\frac{1}{\theta}$ is an integer.

We next describe how $\pushpos$ and
$\poppos$ can be computed in a constant number of (sequential) steps
for tuples with arbitrary $k$, given only the values for tuples with
``special'' values of $k$.
Let $k$ be given and let $a_0,\ldots,a_g$, $a_i\le n^\theta$, for
each $i$ be such
that $k=a_0+a_1n^\theta+\cdots a_gn^{g\theta}$ is the $n^\theta$-adic
representation of $k$. Then $\pushpos((o,u,\tau),k)$ can be determined
as follows: Let $u_{g+1}=u$ and let, for each $i\le g$ with $i\ge 0$, $u_{i}$ be the
prefix of length $\pushpos((o,u_{i+1},\tau),a_in^{i\theta})$. Then $\pushpos((o,u,\tau),k)=u_0$.

Similarly $\poppos$ can be computed for tuples with arbitrary $k$ in a
constant number of (sequential) steps. This computation requires the
computation of values of \pushpos, as well.

Finally, to update \poppos for tuples with special $k$, the
 maximal number $k'$ can be computed in a constant number of stages:
 first, the maximal number of the form $a_gn^{g\theta}$ with the desired
 property is determined, then for the remainder of $v$ the maximal
 number of the form $a_{g-1}n^{(g-1)\theta}$ is determined and so
 on. In
 this way the update per tuple requires only work $\bigO(n^{2\theta})$,
 if we take the actual determination of the maxima into
 account.

 With this modification, only $n^{3+\theta}$ tuples need to be updated
 with work $\bigO(n^{2\theta})$ per tuple. By choosing
 $\theta=\frac{\epsilon}{3}$ the desired work bound
 $\bigO(n^{3+\epsilon})$ can be achieved. 

 Furthermore, all computations are in constant time.

 We finally sketch how insertion operations can be handled. The
 insertion of a symbol $\sigma$ between positions $i$ and $i+1$ of $w$
 can be conceptually handled in two steps: first all positions from
 $i+1$ on are shifted by one to the right and a virtual neutral symbol
 is inserted that does not affect the state nor the stack. All
 relations can be adapted in a straightforward fashion. Since the size
 of the data structures is $\neps[n+3]$, this can be done with work
 $\neps[n+3]$. Then the neutral symbol is replaced by $\sigma$,
 following the procedure from above.

 As soon as
 $|w|$ reaches $\frac{1}{2}n$, the bound $n$ is doubled and the data
 structure is adapted accordingly: it is copied to a larger
 memory area and the ``special values'' for $k$ are adapted. Since function values for arbitrary tuples can be
 computed from function values for special tuples with constant work,
 this transition is possible in constant time and with work in the
 order of the size of the data structure, i.e.,  $\neps[n+3]$.  
\end{proofsketch}
 
\newpage
\section{Appendix for \autoref{subsec:vpl}}
\newcommand{\hleft}{\ensuremath{h_{\swarrow}}\xspace}
\newcommand{\hright}{\ensuremath{h_{\searrow}}\xspace}
A pushdown alphabet is a tuple $\tilde{\Sigma}=(\Sigma_c,\Sigma_r, \Sigma_{\text{int}})$ consisting of three disjoint finite alphabets. We call the symbols in $\Sigma_c$ $\textit{calls}$, the symbols from $\Sigma_r$ \textit{returns} and the symbols from $\Sigma_{\text{int}}$ \textit{internal actions}. 
A Visibly Pushdown automaton works on words over a pushdown
alphabet. 

\begin{definition}[see {\cite[Definition 1]{AlurM04}}]
	\label{definition:vpa}
	A (deterministic) Visibly Pushdown Automaton (VPA) over
        $\tilde{\Sigma}=(\Sigma_c,\Sigma_r, \Sigma_{\text{int}})$ is a tuple
        $\mathcal{A}=(Q,s,\Gamma, \delta, F)$, where $Q$ is a finite
        set of states, $s\in Q$ is the initial state, $\Gamma$ is a
        finite alphabet of stack symbols that contains the special
        bottom-of-stack symbol $\#$. Finally, $\delta$ consists of
        three functions for the three kinds of symbols:
        \begin{itemize}
        \item $\delta_c:  Q \times\Sigma_c \to Q \times (\Gamma \backslash
          \{\#\})$
        \item $\delta_r:  Q\times \Sigma_r \times \Gamma \to Q$
        \item $\delta_{\text{int}} \times
        \Sigma_{\text{int}}\to Q$
        \end{itemize}
      \end{definition}
      We refer to the three different kinds of transitions as push,
      pop and internal transitions, respectively. We denote the language decided by a VPA $\mathcal{A}$ as $L(\mathcal{A})$.

\begin{definition}[\phantom{}{\cite[Definition 2]{AlurM04}}]
	A language $L\subseteq \Sigma^*$ is a Visibly Pushdown Language (VPL) with respect to $\tilde{\Sigma}$ (or $\tilde{\Sigma}$-VPL) if there is a VPA $\mathcal{A}$ over $\tilde{\Sigma}$ such that $L(\mathcal{A})=L$.
\end{definition}
\TheoVPL*

Let $w$ be a string of length $n$ over some pushdown alphabet. For
each position $i$ of $w$, the
\emph{stack height} $s(i)$ of $\calA$ after reading $w[i]$ can be determined 
from the pattern of call and return symbols, as follows. Intitially,
the stack height is 1, that is, $s(0)=1$. 
\[
  s(i) =
  \begin{cases}
    s(i-1)+1 & \text{if $w[i]$ is a call symbol}\\
    s(i-1)-1 & \text{if $w[i]$ is a return symbol and $s(i-1)>1$}\\
    s(i) & \text{in all other cases}
  \end{cases}
\]
We say that positions $i<j$ in $w$ \emph{match}, if $w[i]$ is a call
symbol, $w[j]$ is a return symbol, $s(i)=s(j)$ and for all $k$ with
$i<k<j$, $s(k)>s(i)$ holds.
We note that VPAs can accept words that are not well-formed, that is,
words with unmatched call and retrurn symbols.

Let further the functions
\hleft and \hright be defined as follows, for $i\le n$ and $d\le
n$. They will help the algorithm to jump from a position to the next
position to its left or right, where the stack has $k$ symbols less.
\begin{itemize}
\item If $w[i]$ is a return symbol, then $\hleft(i,d)$ is $j+1$, where
  $j$ is
  the largest position $j<i$ such that
  $s(j)+d=s(i)$. In particular, $\hleft(i,0)$  is the position of the matching call
  symbol for $w[i]$. 
\item If $w[i]$ is a call symbol, then $\hright(i,d)$ is the smallest position $j>i$ such that
  $s(j)+d=s(i)$. We note that $\hright(i,1)$  is the position of the matching return symbol
  for $w[i]$. 
\end{itemize}
Both functions are undefined, if such a position does not exist, or
the precondition does not hold. We
call two positions $i$ and $j$ \emph{matching}, if
$\hright(i,0)=j$ (in which case, also $\hleft(j,0)=i$ holds)

  \begin{proofsketch}
    The algorithm is very similar to the algorithm used for
    \autoref{theo:dcfl}. In fact, it can be observed, that that
    algorithm requires only work  $\bigO(n^{2+\epsilon})$ to maintain
    $\dall$ and $\pushpos$. The only function that requires more work
    is $\poppos$, since the number of tuples to update after a change
    is only bounded by $\bigO(n^{3+\theta})$. The algorithm deals with
    the length bound $w$ in an analogous way as the algorithm for RDCFLs.

    The algorithm for VPLs uses a modified version of \poppos, for
    which it suffices to store $n^{2+2\theta}$ tuples. It is based on
    the following observations.
    \begin{itemize}
    \item In a computation of a VPA  $\calA$ on a word $w$, the height of the stack after
      reading the symbol at
      position $i$ only depends on the pattern of call- and
      return-symbols in $w[1,i]$, but not on the actual
      symbols. Therefore, the strings $y_3$ and $y_2$ in the proof of
      \autoref{theo:dcfl} can be determined from the current word, the
      change position $\ell$ and the types of the old and new symbol
      at $\ell$, without knowledge of any  states  of $\calA$'s run
     on $u_2$.
    \item Likewise, the strings $x_3$ and $x_2$ can be determined
      without knowledge about the states of $\calA$ on $u_1$. 
    \end{itemize}
 Therefore, \popposi need not be stored at all and instead of
 \popposii the algorithm stores values for  the
 function \vpopstate, which is defined as
 \[
   \vpopstate(p,j,\tau,q,m,k) = \dstate(q,w[m,\hright(m,k)],\dstack(p,w[\hleft(j,k),j],\tau)).
 \]
 That is, $\vpopstate(p,j,\tau,q,m,k)$ yields the state that
      $\calA$ reaches in its computation on
            $w[m,\hright(m,k)]$ starting from state $q$ with the stack
            produced by the run of $\calA$ from state $p$ and stack
            $\tau$ on the minimal suffix of $w[j]$ on
            which $\calA$ produces a stack of height $k$.\tsm{Fine-tuning:
        $k$ vs.\ $k+1$ or $k-1$.}
      
      The algorithm stores $\dall$, for all possible
      values, and  $\pushpos$ and 
      $\vpopstate$ for all possible parameter values besides $k$, for
      which only values of the form
$an^{b\theta}$, for integers $b<\frac{1}{\theta}$ and $a\le n^\theta$,
for some fixed $\theta>0$, are considered, similarly as in the proof of
\autoref{theo:dcfl}. Furthermore, it maintains the functions $s$,
$\hleft$ and $\hright$.

It is straighforward to update the function $s$ after a change at
position $\ell$ with $\bigO(n)$ work
per change operation: for each position $i\ge \ell$, $s$ is changed by
at most two, depending on the type of old and new symbol at $\ell$.

Functions, $\hleft$ and $\hright$ can be updated similarly, with work
at most $\bigO(n^2)$. For instance, if a call symbol is replaced by a
return symbol at position $\ell$, reducing $s(\ell)$ by 2, then
\[
  \hleft'(j,k)=
  \begin{cases}
    \hleft(j,k+2) & \text{if $j\ge \ell$ and $\hleft(j,k)\le \ell$,}\\
    \hleft(j,k) & \text{otherwise.}
  \end{cases}
  \]
Here, $\hleft'$ and $\hright'$ denote the updated versions of the
functions.

We do not present the other update algorithms explicitly but rather explain how  they
can be inferred from the
dynamic algorithm for RDCFLs by showing the follows.
\begin{enumerate}[(1)]
\item Values for $\popposii$ can be computed from (not necessarily
  special) values for $\vpopstate$ in constant sequential time.
  \item Values for $\vpopstate$ can be computed from (not necessarily
  special) values for $\popposii$ in constant sequential time.
\item Arbitrary values for $\vpopstate$ can be computed from special
  values for $\vpopstate$.
\end{enumerate}

Therefore, whenever the algorithm for RDCFLs uses $\popposii$, this
use can be replaced by using special values for $\vpopstate$. On the
other hand, updates for $\vpopstate$ can be done by computing suitable
updates values for $\popposii$.

Towards (1), let $C=(p,u,\tau)$,
$u=w[i,j]$ and $v=w[m,n]$, $q$ and $k$ be given. 
Then $\popposii(C,q,v,k)=\vpopstate(p_0,j,\tau_0,q,m,k)$, where
$p_0=\dstate(p,w[i,j_0],\tau)$, $\tau_0=
\topk[1](\dstack(p,w[i,j_0],\tau))$ and $j_0=\hleft(j,k)$. 

Towards (2), $\vpopstate(p,j,\tau,q,m,k)$ is just $\popposii(C,q,
w[m,n],k)$, where $C=(p,\hleft(j,k),k)$.

Finally, we show show (3), how $\vpopstate(p,j,\tau,q,m,k)$ for arbitrary
parameters can be computed from the stored values for $\vpopstate$ in
constant sequential time. This is very similar to the computation of
$\pushpos$ for arbitrary tuples, as described in the proof of \autoref{theo:dcfl}.

 Let us assume again that
$g=\frac{1}{\theta}$ is an integer and let $a_0,\ldots,a_g$, $a_i\le n^\theta$, for
each $i$, be such
that $k=a_0+a_1n^\theta+\cdots a_gn^{g\theta}$ is the $n^\theta$-adic
representation of $k$, as in  the proof of \autoref{theo:dcfl}.

We inductively define sequences $j_0,\ldots,j_g$ and $m_0\ldots,m_g$
of positions, sequences $p_0,\ldots,p_g$ and $r_0,\ldots,r_g$
of states, and a sequence $\tau_0,\ldots,\tau_g$ of stack symbols as
follows.
\begin{itemize}
\item $j_g=j$, $j_{i-1}=\hleft(j_i-1,a_in^{i\theta})$, for each $i\ge 1$;
\item $p_0=p$, $\tau_0=\tau$, 
  $p_{i+1}=\dstate(p_i,w[j_i,j_{i+1}-1],\tau_i)$, $\tau_{i+1}=
  \topk[1](\dstack(p_i,w[j_i,j_{i+1}-1],\tau_i))$, for each $i\ge 0$;
\item $m_g=m$, $m_{i-1}=\hright(m_i,a_in^{i\theta})$, for each $i\ge 1$;
\item $r_g=q$, $r_{i-1}=\vpopstate(p_i,j_i,\tau_i,r_i,m_i,
  a_in^{i\theta})$, for each $i\ge 1$.
\end{itemize}

Therefore, the required work to update a function value of
$\vpopstate$ for one tuple is \ntheta[2] as for RDCFLs. Since there
are \ntheta[2+] tuples, the overall work is \ntheta[2+3] and the
theorem again follows by choosing $\theta=\epsilon$.
Insertions of symbols are handled as for RDCFLs. 
  \end{proofsketch}
 
\end{document}